\documentclass[11pt,a4paper]{article}

\usepackage{amssymb,amsthm,amsmath,mathtools}
\usepackage{thmtools,thm-restate}
\usepackage[round]{natbib}%[numbers]
\usepackage{pifont}
\usepackage[usenames,svgnames,xcdraw,table]{xcolor}
\definecolor{DarkGreen}{rgb}{0.1,0.5,0.1}
\usepackage[backref=page]{hyperref}
\hypersetup{
	colorlinks=true,
	linkcolor=red,%color of internal links
	urlcolor=DarkGreen,%color of external links
	citecolor=blue%color of links to bibliography
}
\renewcommand*{\backref}[1]{}
\renewcommand*{\backrefalt}[4]{%
    \ifcase #1 (Not cited.)%
    \or        (Cited on page~#2)%
    \else      (Cited on pages~#2)%
    \fi}
\usepackage[capitalise,noabbrev]{cleveref}
\Crefname{property}{Property}{Properties}
\Crefname{example}{Example}{Examples}
\Crefname{table}{Table}{Tables}

\usepackage{nicefrac}
\usepackage{enumerate}
\usepackage{etoolbox}%for \ifstrempty{}
\usepackage[margin=1in]{geometry}
\usepackage[T1]{fontenc}
\usepackage[tt=false]{libertine}
\usepackage{url}
\usepackage[utf8]{inputenc}
\usepackage[small]{caption}
\usepackage{graphicx}
\usepackage{booktabs}
\usepackage{mathrsfs,amsfonts,dsfont,authblk}
\usepackage{array,multirow,graphicx,bigdelim}

\usepackage[linesnumbered,lined,boxed,ruled,vlined]{algorithm2e}

\SetKwInOut{Parameters}{Parameters}
\SetKwComment{Comment}{$\triangleright$\ }{}
\SetAlFnt{\small}
\SetAlCapFnt{\small}
\SetAlCapNameFnt{\small}
\SetAlCapHSkip{0pt}
\IncMargin{-\parindent}

\usepackage{tikz}
\usetikzlibrary{fit,calc,shapes,decorations.text,arrows,decorations.markings,decorations.pathmorphing,shapes.geometric,positioning,decorations.pathreplacing}
\tikzset{snake it/.style={decorate, decoration=snake}}

\newcommand*{\tikzmk}[1]{\tikz[remember picture,overlay,] \node (#1) {};\ignorespaces}
\newcommand{\boxit}[1]{\tikz[remember picture,overlay]{\node[yshift=3pt,xshift=4pt,fill=#1,opacity=.25,fit={(A)($(B)+(1.0\linewidth,.8\baselineskip)$)}] {};}\ignorespaces}
\colorlet{mygray}{gray!40}

\makeatletter
\let\oldnl\nl% Store \nl in \oldnl
\newcommand{\nonl}{\renewcommand{\nl}{\let\nl\oldnl}}% Remove line number for one line
\makeatother

  {\list{}{\leftmargin=#1\rightmargin=#1}\item[]}%
  {\endlist}

\usepackage[labelfont={normalfont,bf},textfont=it]{caption}
\usepackage{subcaption}

\theoremstyle{definition}
\newtheorem{assumption}{Assumption}
\newtheorem{examplex}{Example}
\newenvironment{example}{\pushQED{\qed}\examplex}{\popQED\endexamplex}
\theoremstyle{remark}
\newtheorem{remark}{Remark}

\usepackage{flushend}
\usepackage[utf8]{inputenc}
\usepackage[inline]{enumitem}
\Crefname{claim}{Claim}{Claims}

\newtheorem{claim}{Claim}
\allowdisplaybreaks

\newcommand{\AlgEQonePO}{\textup{\textsc{Alg-eq\oldstylenums{1}+po}}}

\newcommand{\DEQ}[1]{\ifstrempty{#1}{\textrm{\textup{DEQ}}}
{\textrm{\textup{DEQ{#1}}}}}
\newcommand{\DEQX}{\textrm{\textup{DEQX}}}
\newcommand{\DEQXzero}{\textrm{\textup{$\DEQX^{0}$}}}
\newcommand{\e}{\mathbf{e}}

\newcommand{\EF}[1]{\ifstrempty{#1}{\textrm{\textup{EF}}}{\textrm{\textup{EF{#1}}}}}
\newcommand{\EFoneone}{${\textrm{\textup{EF}}}^1_1$}
\newcommand{\EFX}{\textrm{\textup{EFX}}}
\newcommand{\EFXzero}{\textrm{\textup{$\EFX^{0}$}}}
\newcommand{\eps}{{\varepsilon}}
\newcommand{\EQ}[1]{\ifstrempty{#1}{\textrm{\textup{EQ}}}
{\textrm{\textup{EQ{#1}}}}}
\newcommand{\EQX}{\textrm{\textup{EQX}}}
\newcommand{\EQXzero}{\textrm{\textup{$\EQX^{0}$}}}
\newcommand{\fPO}{\textrm{\textup{fPO}}}

\newcommand{\I}{\mathcal{I}}
\newcommand{\level}{{\mathrm{level}}}

\newcommand{\Leximin}{\textrm{\textup{Leximin}}}
\newcommand{\M}{{\mathcal{M}}}

\newcommand{\MBB}{\textrm{\textup{MBB}}}
\newcommand{\MMS}{\textrm{\textup{MMS}}}
\newcommand{\N}{{\mathbb{N}}}

\newcommand{\NPC}{\textrm{\textup{NP-complete}}}

\newcommand{\NPH}{\textrm{\textup{NP-hard}}}
\newcommand{\NPhard}{\textup{NP-hard}}

\renewcommand{\O}{{\mathcal{O}}}

\newcommand{\p}{\mathbf p}

\newcommand{\poly}{\textrm{\textup{poly}}}
\newcommand{\PO}{\textrm{\textup{PO}}}

\newcommand{\Prop}[1]{\ifstrempty{#1}{\textrm{\textup{Prop}}}
{\textrm{\textup{Prop{#1}}}}}

\newcommand{\R}{\mathcal{R}}

\newcommand{\ThreePartition}{\textrm{\textsc{3-Partition}}}

\newcommand{\V}{\mathcal{V}}
\newcommand{\VertexCover}{\textrm{\textsc{Vertex Cover}}}
\newcommand{\W}{\mathcal{W}}
\newcommand{\x}{{\mathrm{x}}}

\newcommand{\Z}{\mathbb{Z}}

\let\displaystyle\textstyle

\begin{document}

\title{Equitable Allocations of Indivisible Chores}

\author[1]{Rupert Freeman}
\author[2]{Sujoy Sikdar}
\author[3]{Rohit Vaish}
\author[4]{Lirong Xia}
\affil[1]{Microsoft Research New York City\\
	{\small\texttt{rupert.freeman@microsoft.com}}}
\affil[2]{University of Washington in St. Louis\\
	{\small\texttt{sujoyks@gmail.com}}}
\affil[3]{Rensselaer Polytechnic Institute\\ 
	{\small\texttt{vaishr2@rpi.edu}}}
\affil[4]{Rensselaer Polytechnic Institute\\
	{\small\texttt{xial@cs.rpi.edu}}}

\maketitle

\begin{abstract}
We study fair allocation of indivisible \emph{chores} (i.e., items with non-positive value) among agents with additive valuations. An allocation is deemed \emph{fair} if it is (approximately) \emph{equitable}, which means that the disutilities of the agents are (approximately) equal. Our main theoretical contribution is to show that there always exists an allocation that is simultaneously \emph{equitable up to one chore} (\EQ{1}) and \emph{Pareto optimal} (\PO{}), and to provide a pseudopolynomial-time algorithm for computing such an allocation. In addition, we observe that the Leximin solution---which is known to satisfy a strong form of approximate equitability in the goods setting---fails to satisfy even \EQ{1} for chores. It does, however, satisfy a novel fairness notion that we call \emph{equitability up to any duplicated chore}. Our experiments on synthetic as well as real-world data obtained from the \emph{Spliddit} website reveal that the algorithms considered in our work satisfy approximate fairness and efficiency properties significantly more often than the algorithm currently deployed on Spliddit.
\end{abstract}

\section{Introduction}
\label{sec:Introduction}

Imagine a group of \emph{agents} who must collectively complete a set of undesirable or costly tasks, also known as \emph{chores}. For example, household chores such as cooking, cleaning, and maintenance need to be distributed among the members of the household. As another example, consider the allocation of global climate change responsibilities among the member nations in a treaty~\citep{T02fair}. These responsibilities could entail producing more clean energy, reducing overall emissions, research and development, etc. In both of these cases, it is important that the allocation of chores is \emph{fair} and that it takes advantage of heterogeneity in agents' preferences. For instance, someone might prefer to cook than to clean, while someone else might have the opposite preference. Likewise, different countries might have competitive advantages in different areas.

Problems of this nature can be modeled mathematically as \emph{chore division} problems, first introduced by \citet{G78aha}. Each agent incurs a non-positive utility, or cost (in terms of money, time, or general dissatisfaction), from completing each chore that she reports to a central mechanism. In this paper, our focus is on designing mechanisms to divide the chores among the agents \emph{equitably}. An allocation of chores is equitable if all agents get exactly the same (dis)utility from their allocated chores. Other fairness properties can, of course, be considered too---for instance, \emph{envy-freeness} dictates that no agent should prefer another agent's assigned chores to her own. While this is not the main focus of our work, we do consider (approximate) envy-freeness in conjunction with (approximate) equitability.

Equitable allocations have been studied extensively in the context of allocating \emph{goods} (i.e., items with non-negative value). When the goods are divisible (or, even more generally, in the \emph{cake-cutting} setting), perfectly equitable allocations are guaranteed to exist~\citep{DS61cut,A87splitting}. For indivisible goods, though, perfect equitability might not be possible, but approximate versions can still be achieved~\citep{GMT14near,FSV+19equitable}.

At first glance, the problem of chore division appears similar to the goods division problem. However, there are subtle technical differences between the two settings. In the context of (approximate) envy-freeness, this contrast has been noted in several works~\citep{PS09nperson,BMS+18dividing,BMS+17competitive,BS19algorithms}. To take one example, it is known that an allocation of goods that is both \emph{envy-free up to one good} and Pareto optimal can be found by allocating goods so that the \emph{product} of the agents' utilities---the Nash social welfare---is maximized~\citep{CKM+19unreasonable}. However, maximizing the product of utilities is not sensible when valuations are negative, and no analogous procedure is known for the case of chores.

\begin{table}%[ht]
	\footnotesize
	\centering
	\begin{tabular}{
>{\centering}m{0.1\textwidth}|
>{\centering}m{0.1\textwidth}|
>{\centering}m{0.15\textwidth}
>{\centering}m{0.15\textwidth}
>{\centering}m{0.15\textwidth}
%>{\centering}m{0.13\textwidth}
>{\centering\arraybackslash}m{0.15\textwidth}}
		\toprule
		\multicolumn{2}{c}{} & \textbf{\EQ{1}} & \textbf{\EQX{}} & \textbf{\DEQ{1}} & \textbf{\DEQX{}}\\
		\cmidrule{1-6}
		\scriptsize
		\multirow{4}{*}{\textbf{Without \PO{}}} & \multirow{2}{*}{Existence} & Always exists & Always exists & Always exists & Always exists\\
		& & (\Cref{prop:EQX_Existence_Computation}) & (\Cref{prop:EQX_Existence_Computation}) & (\Cref{prop:DEQX+PO})& (\Cref{prop:DEQX+PO})\\
		\cmidrule{2-6}
		& \multirow{2}{*}{Computation} & Poly time & Poly time & Poly time & \multirow{2}{*}{?} \\
		& & (\Cref{prop:EQX_Existence_Computation}) & (\Cref{prop:EQX_Existence_Computation}) & (\Cref{prop:dEQ1_Computation}) &\\
		\cmidrule{1-6}
		\cmidrule{1-6}
		\cmidrule{1-6}
		\multirow{4}{*}{\textbf{With \PO{}}} & \multirow{2}{*}{Existence} & Always exists & Might not exist & Always exists & Always exists\\
		& & (\Cref{thm:EQ1+PO_pseudopolynomial}) & (\Cref{eg:Nonexistence_EQx+PO_Two_Agents}) & (\Cref{prop:DEQX+PO})& (\Cref{prop:DEQX+PO})\\
		\cmidrule{2-6}
		& \multirow{2}{*}{Computation} & Pseudopoly time & Strongly \NPH{} & \multirow{2}{*}{?} & \multirow{2}{*}{?} \\
		& & (\Cref{thm:EQ1+PO_pseudopolynomial}) & (\Cref{thm:EQX_PO_Strong_NP-hardness}) & &\\
		\bottomrule
	\end{tabular}
	\caption{Summary of our theoretical results. Each cell contains the existence/computation results for various combinations of fairness and efficiency properties. Open questions are marked with `?'.}
	\label{tab:Results}
\end{table}

In this paper, we demonstrate a similar set of differences between the goods and chores settings in the context of equitability. Our focus is on \emph{equitability up to one/any chore} (\EQ{1}/\EQX{}) which requires that pairwise violations of equitability can be eliminated by removing some/any chore from the bundle of the less happy agent.

For goods division, \citet{FSV+19equitable} showed that equitability up to any good and Pareto optimality are achieved simultaneously by the Leximin algorithm.\footnote{The Leximin algorithm maximizes the utility of the least well-off agent, and subject to that maximizes the utility of the second-least, and so on.} However, we show that in the chores setting, \emph{Leximin does not even guarantee equitability up to one chore} (\EQ{1}) (\Cref{eg:Leximin_fails_EQl_and_EF1}). Further, while we are able to give \emph{an algorithm that outputs an \EQ{1} and \PO{} allocation in pseudopolynomial time} (\Cref{thm:EQ1+PO_pseudopolynomial}), modifying a similar algorithm of~\citet{FSV+19equitable}, we show that \emph{an allocation satisfying \EQX{} and \PO{} may not exist}, in contrast to the goods setting (\Cref{eg:Nonexistence_EQx+PO_Two_Agents}). 

The fact that \EQX{}+\PO{} could fail to exist and that the Leximin allocation may not be \EQ{1} leads us to consider other relaxations of perfect equitability. To this end, we define the \emph{equitability up to one/any duplicated chore} (\DEQ{1}/\DEQX{}) properties. These properties require that pairwise equitability can be restored by duplicating a chore from the less happy agent's bundle and adding it to the more happy agent's bundle, rather than removing a chore from the less happy agent's bundle. Interestingly, we find that the ``duplicate'' relaxations are satisfied by the Leximin allocation (\Cref{prop:DEQX+PO}), restoring a formal justification for that algorithm even in the chores setting. \Cref{tab:Results} summarizes our results.

Finally, we complement our theoretical results with extensive simulations on both simulated data and data gathered from the popular fair division website \emph{Spliddit} ~\citep{GP15spliddit}.\footnote{\texttt{\url{http://www.spliddit.org/apps/tasks}}} We find that on a large fraction of instances ($>80\%$), Leximin satisfies all of the approximate properties that we consider, in addition to Pareto optimality. We therefore consider it to be the best choice for use in practice, matching the observation of~\citet{FSV+19equitable} in the case of goods. 
When the runtime of the Leximin algorithm is prohibitive (computing the Leximin allocation is \NPhard), our simulations reveal that our pseudopolynomial algorithm for achieving \EQ{1} and \PO{} is a reasonable choice for achieving these as well as other properties on a large fraction of instances.

\subsection{Related Work}
\label{subsec:RelatedWork}

Fair division of indivisible chores has received considerable interest in recent years. \citet{ARS+17algorithms}, \citet{HL19algorithmic}, \citet{ACL19weighted}, and \citet{ALW19strategyproof} study approximation algorithms for max-min fair share (\MMS{}) allocation of chores.
\citet{BS19algorithms} show that an allocation that is envy-free up to the removal of two chores (\EFoneone{}) and Pareto optimal (\PO{}) always exists and can be computed in polynomial time if the number of agents is fixed. \citet{SH18competitive} has studied competitive equilibria in the allocation of indivisible chores with unequal budgets.

Several papers study a model with \emph{mixed} items, wherein an item can be a good for one agent and a chore for another. \citet{BMS+17competitive} examine this model for \emph{divisible} items and show that unlike the goods-only case, the set of competitive utility profiles~\citep{V74equity,EG59consensus} can be multivalued; for the chores-only problem, the multiplicity can be exponential in the number of agents/items~\citep{BMS+18dividing}. \citet{S18fairly} and \citet{MZ19envy} consider a generalization of the cake-cutting problem to the mixed utilities setting, and study envy-free divisions with connected pieces. \citet{ACI+19fair} study \emph{indivisible} mixed items and provide a polynomial-time algorithm for computing \EF{1} allocations even for non-additive valuations. For the same model, \citet{AMS19polynomial} provide a polynomial-time algorithm for computing allocations that are Pareto optimal (\PO{}) and proportional up to one item (\Prop{1}). \citet{SS19fair} consider envy-free/proportional and Pareto optimal divisions that minimize the number of fractionally assigned items. Notably, none of this work examines equitability for indivisible items.

Equitability for indivisible chores has been studied by \citet{BCL19chore} in a model where the items constitute the vertices of a graph, and each agent should be assigned a connected subgraph. This work does not consider Pareto optimality, and the space of permissible allocations in this model is different from ours, making the two sets of results incomparable. \citet{CKK+12efficiency} study the worst-case welfare loss due to equitability (i.e., `price of equitability') for indivisible chores, but do not consider approximate fairness.

For goods, equitability as a fairness notion has been studied extensively, mostly in the context of cake-cutting~\citep{DS61cut,BJK06better,CP12computability,BFL+12maxsum,CDP13existence,BJK13n,AD15efficiency,PW17lower,C17existence}. Our work bears most similarity to the work of \citet{GMT14near} and \citet{FSV+19equitable}, who define the notions of \EQX{} and \EQ{1}, respectively.
\section{Preliminaries}
\label{sec:Preliminaries}

\paragraph{Problem instance}
An \emph{instance} $\langle [n], [m], \V \rangle$ of the fair division problem is defined by a set of $n \in \N$ \emph{agents} $[n] = \{1,2,\dots,n\}$, a set of $m \in \N$ \emph{chores} $[m] = \{1,2,\dots,m\}$, and a \emph{valuation profile} $\V = \{v_1,v_2,\dots,v_n\}$ that specifies the preferences of every agent $i \in [n]$ over each subset of the chores in $[m]$ via a \emph{valuation function} $v_i: 2^{[m]} \rightarrow \Z_{\leq 0}$. Note that we assume that the valuations are non-positive integers; most of our results hold without this assumption but \Cref{thm:EQ1+PO_pseudopolynomial} requires it.

We will also assume that the valuation functions are \emph{additive}, i.e., for any agent $i \in [n]$ and any set of chores $S \subseteq [m]$, $v_i(S) \coloneqq \sum_{j \in S} v_i(\{j\})$, where $v_i(\emptyset) = 0$. For a singleton chore $j \in [m]$, we will write $v_{i,j}$ instead of $v_i(\{j\})$. The valuation functions are said to be \emph{normalized} if for all agents $i,j \in [n]$, we have $v_i([m]) = v_j([m])$. We will assume throughout, without loss
of generality, that for each chore $j \in [m]$, there exists some agent $i \in [n]$ with a non-zero valuation for
it (i.e., $v_{i,j} < 0$), and for each agent $i \in [n]$, there exists a chore $j \in [m]$ that it has non-zero value for.

\paragraph{Allocation}
An \emph{allocation} $A \coloneqq (A_1,\dots,A_n)$ is an $n$-partition of the set of chores $[m]$, where $A_i \subseteq [m]$ is the \emph{bundle} allocated to the agent $i$ (note that $A_i$ can be empty). Given an allocation $A$, the \emph{utility} of agent $i \in [n]$ for the bundle $A_i$ is $v_i(A_i) = \sum_{j \in A_i} v_{i,j}$.

\paragraph{Equitability}
An allocation $A$ is said to be (a) \emph{equitable} (\EQ) if for every pair of agents $i,k \in [n]$, we have $v_i(A_i) = v_k(A_k)$; (b) \emph{equitable up to one chore} (\EQ1) if for every pair of agents $i,k \in [n]$ such that $A_i \neq \emptyset$, there exists a chore $j \in A_i$ such that $v_i(A_i \setminus \{j\}) \geq v_k(A_k)$, and (c) \emph{equitable up to any chore} (\EQX{}) if for every pair of agents $i,k \in [n]$ such that $A_i \neq \emptyset$ and for every chore $j \in A_i$ such that $v_{i,j} < 0$, we have $v_i(A_i \setminus \{j\}) \geq v_k(A_k)$. These notions have been previously studied for goods by \citet{GMT14near} and \citet{FSV+19equitable}. Our presentation of the notions of (approximate) equitability for chores---in particular, the idea of removing a chore from the less-happy agent's bundle---follows the formulation used by \citet{ACI+19fair} and \citet{A18almost} in defining the analogous relaxations of envy-freeness (see below).

\paragraph{Envy-freeness}
An allocation $A$ is said to be (a) \emph{envy-free} (\EF{}) if for every pair of agents $i,k \in [n]$, we have $v_i(A_i) \geq v_i(A_k)$; (b) \emph{envy-free up to one chore} (\EF{1}) if for every pair of agents $i,k \in [n]$ such that $A_i \neq \emptyset$, there exists a chore $j \in A_i$ such that $v_i(A_i \setminus \{j\}) \geq v_i(A_k)$, and (c) \emph{envy-free up to any chore} (\EFX{}) if for every pair of agents $i,k \in [n]$ such that $A_i \neq \emptyset$ and for every chore $j \in A_i$ such that $v_{i,j} < 0$, we have $v_i(A_i \setminus \{j\}) \geq v_i(A_k)$. The notions of \EF{}, \EF{1}, and \EFX{} were proposed in the context of goods allocation by \citet{F67resource}, \citet{B11combinatorial}, and \citet{CKM+19unreasonable}, respectively.\footnote{An earlier work by \citet{LMM+04approximately} studied a weaker approximation of envy-freeness for goods, but their algorithm is known to compute an \EF{1} allocation.}

It is easy to see that envy-freeness and equitability (and their corresponding relaxations) become equivalent when the valuations are \emph{identical}, i.e., for every $j \in [m]$, $v_{i,j} = v_{k,j}$ for all $i,k \in [n]$.

\begin{restatable}{prop}{IdenticalValuations}
 \label{prop:IdenticalValuations}
When agents have identical valuations, an allocation satisfies \EF{}/\EFX{}/\EF{1} if and only if it satisfies \EQ{}/\EQX{}/\EQ{1}.
\end{restatable}

\paragraph{Pareto optimality}
An allocation $A$ is Pareto dominated by allocation $B$ if $v_k(B_k) \geq v_k(A_k)$ for every agent $k \in [n]$ with at least one of the inequalities being strict. A \emph{Pareto optimal} (\PO{}) allocation is one that is not Pareto dominated by any other allocation.

\paragraph{Leximin-optimal allocation}
A \Leximin{}-optimal (or simply Leximin) allocation of chores is one that maximizes the minimum utility that any agent achieves, subject to which the second minimum utility is maximized, and so on. The utilities induced by a \Leximin{} allocation are unique, although there may exist more than one such allocation~\citep{DS61cut}.
\section{Theoretical Results}
\label{sec:Theoretical_Results}

This section presents our theoretical contributions. We will first consider equitability and its relaxations (\Cref{subsec:EQ}), followed by combining these notions with Pareto optimality (\Cref{subsec:EQ_PO}), and subsequently also considering envy-freeness (\Cref{subsec:EQ_PO_EF}). Finally, we will discuss a novel approximation of equitability called \emph{equitability up to a duplicated chore} (\Cref{subsec:Dup_EQ}).

\subsection{Equitability and its Relaxations}
\label{subsec:EQ}

As discussed previously, an equitable (\EQ{}) allocation is not guaranteed to exist when allocating indivisible chores. In addition, the computational problem of determining whether a given instance has an equitable allocation turns out to be \NPC{} even for identical valuations (\Cref{prop:EQNPhardIdenticalVals}). The proof uses a standard reduction from \ThreePartition{} and is therefore omitted.

\begin{restatable}{prop}{EQNPhardIdenticalVals}
 \label{prop:EQNPhardIdenticalVals}
Determining whether a given fair division instance admits an equitable $(\EQ{})$ allocation is strongly \NPC{} even for identical valuations.
\end{restatable}

The negative results regarding the existence and computation of exact equitability are in complete contrast with those of its relaxations. Indeed, when allocating indivisible chores, there always exists an allocation that is equitable up to any chore (\EQX{}). Furthermore, such an allocation can be computed in polynomial time via a simple greedy procedure (\Cref{prop:EQX_Existence_Computation}). This algorithm is a straightforward adaptation of the algorithm of \citet{GMT14near} for computing \EQX{} allocations of goods.

\begin{restatable}[]{prop}{EQXExistenceComputation}
 \label{prop:EQX_Existence_Computation}
 An \EQX{} allocation of chores always exists and can be computed in polynomial time.
\end{restatable}
\begin{proof} (Sketch)
Our algorithm iteratively assigns the chores to the agents according to the following assignment rule: At each step, the \emph{happiest} agent (i.e., one whose utility is closest to zero) is asked to select a chore from the set of available chores that it \emph{dislikes} the most (i.e., the chore that gives it the most negative utility).

It is easy to see that the chore assigned most recently to any agent is also its favorite (or least disliked) chore in its bundle. Thus, if an allocation is \EQX{} before assigning a chore, then it continues to be \EQX{} after it. The claim now follows by induction, since an empty allocation is \EQX{} to begin with.
\end{proof}

The positive result in \Cref{prop:EQX_Existence_Computation} offers an interesting comparison between envy-freeness and equitability: It is not known whether an \EFX{} allocation is even guaranteed to exist for chores, but an \EQX{} allocation can always be computed in polynomial time. 

\subsection{Equitability and Pareto Optimality}
\label{subsec:EQ_PO}

We will now consider equitability together with Pareto optimality. From \Cref{prop:EQNPhardIdenticalVals}, it is easy to see that checking the existence of an equitable and Pareto optimal (\EQ{}+\PO{}) allocation is strongly \NPH{} (since every allocation is Pareto optimal under identical valuations). Therefore, we will strive for achieving Pareto optimality alongside approximate equitability, specifically \EQ{1} and \EQX{}.

We will start by considering equitability up to any chore (\EQX{}) and Pareto optimality. For goods allocation, \citet{FSV+19equitable} have shown that equitability up to any good and Pareto optimality can be simultaneously achieved using the Leximin allocation.\footnote{This result requires the valuations to be strictly positive.} By contrast, as we show in \Cref{eg:Nonexistence_EQx+PO_Two_Agents}, there might not exist an allocation that is equitable up to any chore and Pareto optimal, even when there are only two agents.

\begin{example}[\textbf{Non-existence of \EQX{}+\PO{}}]
\label{eg:Nonexistence_EQx+PO_Two_Agents}
Consider an instance with three chores $c_1,c_2,c_3$ and two agents $a_1,a_2$ with strictly negative (and normalized) valuations as shown below:
\begin{table}[ht]
\centering
\begin{tabular}{ c|ccc }
	& $c_1$ & $c_2$ & $c_3$\\ \hline
  $a_1$ & $-2$ & $-50$ & $-50$\\
  $a_2$ & $-97$ & $-4$ & $-1$
\end{tabular}
\end{table}

Of the eight possible allocations in the above instance, the two allocations that assign all chores to a single agent, namely $A^{1} \coloneqq (\{c_1,c_2,c_3\},\{\emptyset\})$ and $A^{2} \coloneqq (\{\emptyset\},\{c_1,c_2,c_3\})$ violate \EQX{} and can be immediately ruled out. Any other allocation must assign exactly one chore to one agent and two to the other.

Of the three allocations in which $a_1$ is assigned exactly one chore, namely $A^3 \coloneqq (\{c_1\},\{c_2,c_3\})$, $A^4 \coloneqq (\{c_2\},\{c_1,c_3\})$, and $A^5 \coloneqq (\{c_3\},\{c_1,c_2\})$, none satisfies \EQX{}. Therefore, these allocations can be ruled out as well.

This leaves us with the three allocations in which $a_2$ is assigned exactly one chore, namely $A^6 \coloneqq (\{c_1,c_2\},\{c_3\})$, $A^7 \coloneqq (\{c_2,c_3\},\{c_1\})$, and $A^8 \coloneqq (\{c_1,c_3\},\{c_2\})$. Among these, only $A^7$ satisfies \EQX{}. However, $A^7$ is Pareto dominated by the allocation $A^3$; indeed, $v_1(A^7) = -100 < v_1(A^3) = -2$ and $v_2(A^7) = -97 < v_2(A^3) = -5$. Therefore, the above instance does not admit an \EQX{}+\PO{} allocation.
\end{example}

To make matters worse, determining whether a given instance admits an \EQX{} and \PO{} allocation turns out to be strongly \NPH{}.

\begin{restatable}[\textbf{Strong \NPH{}ness of \EQX{}+\PO{}}]{theorem}{EQXPOStrongNPhardness}
 \label{thm:EQX_PO_Strong_NP-hardness}
 Determining whether a given fair division instance admits an allocation that is simultaneously equitable up to any chore $(\EQX{})$ and Pareto optimal $(\PO{})$ is \textup{strongly \NPH{}}, even for strictly negative and normalized valuations.
\end{restatable}
\begin{proof}
We will show a reduction from \ThreePartition{}, which is known to be strongly \NPhard{} \citep[Theorem 4.4]{GJ79computers}. An instance of \ThreePartition{} consists of a set $X = \{b_1,\dots,b_{3r}\}$ of $3r$ positive integers where $r \in \N$, and the goal is to find a partition of $X$ into $r$ subsets $X^1,\dots,X^r$ such that the sum of numbers in each subset is equal to $B \coloneqq \frac{1}{r} \sum_{b_i \in X} b_i$.\footnote{Note that we do not require the sets $X^1,\dots,X^r$ to be of cardinality three each; \ThreePartition{} remains strongly \NPhard{} even without this constraint.} We will assume, without loss of generality, that for every $i \in [3r]$, $b_i$ is even and $b_i \geq 2$. As a result, we can also assume, without loss of generality, that $B$ is even.

We will construct a fair division instance with $r+1$ agents and $4r+2$ chores (see \Cref{tab:EQX+PO_reduced_instance}). The set of agents consists of $r$ \emph{main} agents $a_1,\dots,a_r$ and a \emph{dummy} agent $d$. The set of chores consists of $3r$ \emph{main} chores $C_1,\dots,C_{3r}$, $r$ \emph{signature} chores $S_1,\dots,S_r$, and two \emph{dummy} chores $D_1,D_2$. The valuations of the agents are specified as follows: For every $i \in [r]$ and $j \in [3r]$, agent $a_i$ values the main chore $C_j$ at $-b_j$, the signature chore $S_i$ at $-1$, and all other chores at a large negative number $-L$, where $-L<-rB-1$. The dummy agent $d$ values the dummy chores $D_1$ and $D_2$ at $-1$ and $-B$, respectively, and all other chores at a large negative number $-L'$. In the interest of having normalized valuations, we can set $L' \coloneqq \frac{(r-1)B + (r+1)L}{4r}$. It is easy to show using standard calculus that $-L' < -B$ for all $r \geq 3$. Since the condition $r \geq 3$ holds without loss of generality, we will assume throughout that $-L' < -B$.
\begin{table}[t]
\centering
%\small
\begin{tabular}{ c||ccc|cccc|cc }
& $C_1$ & \ldots & $C_{3r}$ & $S_1$ & $S_2$ & \ldots & $S_{r}$ & $D_1$ & $D_2$ \\
\hline
\hline
$a_1$ & $-b_1$ & \ldots & $-b_{3r}$ & $-1$ & $-L$ & \ldots & $-L$ & $-L$ & $-L$ \\
$a_2$ & $-b_1$ & \ldots & $-b_{3r}$ & $-L$ & $-1$ & \ldots & $-L$ & $-L$ & $-L$ \\
\vdots && \vdots &&& \vdots && &\vdots &\\
$a_r$ & $-b_1$ & \ldots & $-b_{3r}$ & $-L$ & $-L$ & \ldots & $-1$ & $-L$ & $-L$ \\
$d$ & $-L'$ & \ldots & $-L'$ & $-L'$ & $-L'$ & \ldots & $-L'$ & $-1$ & $-B$\\
\end{tabular}
\caption{Chores instance used in the proof of \Cref{thm:EQX_PO_Strong_NP-hardness}.}
\label{tab:EQX+PO_reduced_instance}
\end{table}

We will now argue the equivalence of solutions.

$(\Rightarrow)$ Let $X^1,\dots,X^r$ be a solution of \ThreePartition{}. Then, the desired allocation $A$ can be constructed as follows: For every $i \in [r]$, the main agent $a_i$ gets the signature chore $S_i$ as well as the chores corresponding to the numbers in $X^i$. The dummy agent gets the two dummy chores. The allocation $A$ is Pareto optimal because each chore is assigned to an agent that has the highest valuation for it (thus, $A$ maximizes social welfare). Also, each agent's valuation in $A$ is $-B-1$, implying that $A$ is equitable, and hence also \EQX{}.

$(\Leftarrow)$ Now suppose that there exists an \EQX{} and Pareto optimal allocation $A$. Below, we will make a series of observations about $A$ that will help us infer a solution of \ThreePartition{} using $A$.

\begin{claim}
No agent gets an empty bundle in $A$.
\label{claim:No_empty_bundles}
\end{claim}
\begin{proof} (of \Cref{claim:No_empty_bundles})
If an agent gets an empty bundle, then some other agent will get four or more chores (as more than $4r$ chores will need to be allocated among $r$ other agents). Since all valuations are strictly negative, this results in a violation of \EQX{}.
\end{proof}

\begin{claim}
Each main agent $a_i$ gets its signature chore $S_i$ in $A$.
\label{claim:Signature_chores}
\end{claim}
\begin{proof} (of \Cref{claim:Signature_chores})
From \Cref{claim:No_empty_bundles}, we know that $a_i$ owns at least one chore in $A$. Fix any chore $j \in A_i$. Suppose $S_i$ is assigned to an agent $k$ in $A$. Notice that the valuation of agent $k$ for $S_i$ is either $-L$ or $-L'$ (depending of whether $k$ is a main or a dummy agent). This is also the smallest valuation that agent $k$ has for \emph{any} chore (recall that $-L < -rB-1$ and $-L' < -B$). Furthermore, since $-b_i \leq -2$ for every $i \in [3r]$, $S_i$ is the unique favorite chore of agent $a_i$. Therefore, after exchanging the chores $j$ and $S_i$, the valuation of agent $k$ cannot decrease (due to additivity), and the valuation of agent $a_i$ necessarily increases. Thus, the new allocation is a Pareto improvement over $A$, which is a contradiction.
\end{proof}

\begin{claim}
The chore $D_1$ is assigned to the dummy agent $d$ in $A$.
\label{claim:Dummy_agent_gets_D1}
\end{claim}
\begin{proof} (of \Cref{claim:Dummy_agent_gets_D1})
By an argument similar to that in the proof of \Cref{claim:Signature_chores}, we can show that if $D_1$ is not assigned to $d$, then a Pareto improving swap between $d$ and the owner of $D_1$ is possible.
\end{proof}

\begin{claim}
The chore $D_2$ is assigned to the dummy agent $d$ in $A$.
\label{claim:Dummy_agent_gets_D2}
\end{claim}
\begin{proof} (of \Cref{claim:Dummy_agent_gets_D2})
Suppose, for contradiction, that $D_2$ is assigned to main agent $a_i$ in $A$. From \Cref{claim:Signature_chores}, we know that $a_i$ is also assigned its signature chore $S_i$. Since $S_i$ is the favorite chore of $a_i$, the \EQX{} condition requires that for every other main agent $a_k$, 
$$v_k(A_k) \leq v_i(A_i \setminus \{S_i\}) \leq v_i(\{D_2\}) = -L.$$

Even if $a_k$ is assigned all the remaining chores whose assignment has not been finalized yet (this includes the $3r$ main chores), its valuation will still only be $-rB-1 > -L$. This would imply a violation of \EQX{} condition between $a_i$ and $a_k$, which is a contradiction.
\end{proof}

From \Cref{claim:Dummy_agent_gets_D1,claim:Dummy_agent_gets_D2}, we know that $D_1,D_2 \in A_d$. Therefore, by \EQX{} condition, the following must hold for every main agent $a_i$:
$$v_i(A_i) \leq v_d(A_d \setminus \{D_1\}) \leq v_d(\{D_2\}) = -B.$$

From \Cref{claim:Signature_chores}, we know that $a_i$ gets its signature chore $S_i$. Thus, the valuation of $a_i$ for the remaining items in its bundle must be
\begin{equation}
v_i(A_i \setminus \{S_i\}) \leq -B+1.
\label{eqn:Valuation_Bound}
\end{equation}

Since the assignment of all signature and dummy chores has been fixed, the set $A_i \setminus \{S_i\}$ can only have main chores. By assumption, main agents have even-valued valuations for main chores. By additivity of valuations, the quantity $v_i(A_i \setminus \{S_i\})$ must also be even. Also, $-B$ is even, so $-B+1$ must be odd, and therefore the inequality in \Cref{eqn:Valuation_Bound} must be strict. Thus, $v_i(A_i \setminus \{S_i\}) \leq -B$.

We can now infer a solution of \ThreePartition{} as follows: For every $i \in [r]$, the set $X^i$ contains those numbers whose corresponding chores are included in $A_i \setminus \{S_i\}$. Since $v_i(A_i \setminus \{S_i\}) \leq -B$, it follows that all main chores must be assigned among the main agents, implying that $X^1,\dots,X^r$ constitute a valid partition of $X$. Furthermore, the sum of numbers in the set $X^i$ cannot exceed $B$, or otherwise the sum of numbers in some other set $X^k$ will be strictly less than $B$, which would violate the above inequality. Hence, $X^1,\dots,X^r$ is a valid solution of \ThreePartition{}, as desired.
\end{proof}

The negative results concerning the existence and computation of \EQX{}+\PO{} lead us to consider a weaker relaxation of equitability, namely equitability up to one chore (\EQ{1}). A natural starting point in studying the existence of \EQ{1}+\PO{} allocations is the Leximin solution, as it yields strong positive results for the goods setting~\citep{FSV+19equitable}. Unfortunately, as \Cref{eg:Leximin_fails_EQl_and_EF1} shows, Leximin sometimes fails to satisfy \EQ{1} (as well as \EF{1}) for chores.

\begin{example}[\textbf{\Leximin{} fails \EQ{1} and \EF{1}}]
\label{eg:Leximin_fails_EQl_and_EF1}
Consider the following instance with four chores and three agents with normalized and strictly negative valuations:
\begin{table}[ht]
\centering
\begin{tabular}{ c|ccccc }
& $c_1$ & $c_2$ & $c_3$ & $c_4$\\ \hline
  $a_1$ & $-1$ & $-5$ & $-5$ & $-5$\\
  $a_2$ & $-1$ & $-2$ & $-2$  & $-11$\\
  $a_3$ & $-6$ & $-5$ & $-3$ & $-2$\\
\end{tabular}
\end{table}

We will show that the allocation $A$ given by $A_1 = \{c_1\}$, $A_2 = \{c_2,c_3\}$, and $A_3 = \{c_4\}$ is \Leximin{}-optimal. Suppose, for contradiction, that another allocation $B$ is a Leximin improvement over $A$. The utility profile induced by $A$ is $(-1,-4,-2)$, and therefore, for any chore $j$ and agent $i$ such that $j \in B_i$, we must have that $v_{i,j} \geq -4$.

The chore $c_4$ is valued at less than $-4$ by both $a_1$ and $a_2$, so we must have $c_4 \in B_3$. Similarly, we can fix $c_2 \in B_2$. This, in turn, forces us to fix $c_3 \in B_2$, since otherwise if $c_3 \in B_3$, then the utility of $a_3$ will be $-5 < -4$, which would violate the \Leximin{} improvement assumption. By a similar argument, we have $c_1 \in B_1$. This, however, implies that $A$ and $B$ are identical, which is a contradiction. Therefore, $A$ must be \Leximin{}. Notice that $A$ violates \EQ{1} and \EF{1} for the pair $(a_1,a_2)$.
\end{example}

Another natural approach to show the existence of \EQ{1}+\PO{} allocations could be to use the \emph{relax-and-round} framework. Specifically, one could start from an \emph{egalitarian-equivalent} solution~\citep{PS78egalitarian} (i.e., a fractional allocation that is perfectly equitable and minimizes the agents' disutilities), and use a rounding algorithm to achieve \EQ{1}. However, there is a simple example where this approach fails.\footnote{Consider an instance with three chores and three agents. Agents 1 and 2 value the first chore at $-4$ and the other two chores at $-\infty$ (or a suitably large negative value). Agent 3 values the first chore at $-\infty$ and the other two chores at $-1$ each. An egalitarian-equivalent solution divides the first chore equally between agents 1 and 2, and assigns the remaining two chores to agent 3. Any rounding of this fractional allocation violates \EQ{1} with respect to agent 3 and whoever of agents 1 or 2 gets an empty bundle.}

The failure of Leximin and the relax-and-round framework in achieving \EQ{1} prompts us to consider a different approach for studying approximately fair and Pareto optimal allocations. This approach, which is based on Fisher markets~\citep{BS00compute}, has been successfully used in the goods model to provide an algorithmic framework for computing \EF{1}+\PO{}~\citep{BKV18Finding} and \EQ{1}+\PO{}~\citep{FSV+19equitable} allocations.\footnote{Similar techniques have also been used in developing approximation algorithms for Nash social welfare objective for budget-additive and multi-item concave utilities~\citep{CCG+18fair}.}
 Note that the existence of such allocations was established by means of computationally intractable methods, namely the Maximum Nash Welfare and Leximin solutions~\citep{CKM+19unreasonable,FSV+19equitable}.

Briefly, the idea is to start with an allocation that is an equilibrium of some Fisher market. By the first welfare theorem~\citep{MWG+95microeconomic}, such an allocation is guaranteed to be Pareto optimal. By using a combination of local search and price change steps, our algorithm converges to an \emph{approximately} equitable equilibrium, which gives an approximately equitable and Pareto optimal allocation. An important distinguishing feature of our algorithm is that while the existing Fisher market based approaches use \emph{price-rise}~\citep{BKV18Finding,FSV+19equitable}, our algorithm instead uses \emph{price-drop} as the natural option for negative valuations.

Our main result in this section (\Cref{thm:EQ1+PO_pseudopolynomial}) establishes the existence of \EQ{1} and \PO{} allocations using the markets framework.

\begin{restatable}[\textbf{Algorithm for \EQ{1}+\PO{}}]{theorem}{EQonePOPseudopolynomial}
 \label{thm:EQ1+PO_pseudopolynomial}
 Given any chores instance with additive and integral valuations, an allocation that is equitable up to one chore $(\EQ{1})$ and Pareto optimal $(\PO{})$ always exists and can be computed in $\O(\poly(m,n,|v_{\min}|))$ time, where $v_{\min} = \min_{i,j} v_{i,j}$.
\end{restatable}

In particular, when the valuations are polynomially bounded (i.e., for every $i \in [n]$ and $j \in [m]$, $v_{i,j} \leq \poly(m,n)$), our algorithm computes an \EQ{1} and \PO{} allocation in \emph{polynomial time}. Whether an \EQ{1}+\PO{} allocation can be computed in polynomial time without this assumption is an interesting avenue for future research.\footnote{Interestingly, similar questions concerning the computation of \EF{1}+\PO{} or \EQ{1}+\PO{} allocations are also open in the goods setting~\citep{BKV18Finding,FSV+19equitable}.}

The proof of \Cref{thm:EQ1+PO_pseudopolynomial} is deferred to \Cref{subsec:EQ1+PO}. Here, we will provide an informal overview of the algorithm by demonstrating its execution on the instance in \Cref{eg:Leximin_fails_EQl_and_EF1} where Leximin fails to satisfy \EQ{1}.

\begin{figure}
\scriptsize
\centering
\begin{subfigure}[b]{0.2\textwidth}
%\subfloat[][]{
\begin{tikzpicture}
	\tikzset{agent/.style = {shape=circle,draw,inner sep=1pt}}
	\tikzset{agentutility/.style = {shape=circle,inner sep=0pt}}
	\tikzset{chore/.style = {shape=circle,draw,inner sep=1pt}}
	\tikzset{choreprice/.style = {shape=circle,inner sep=0pt}}
	\tikzset{edge/.style = {solid,line width=1.2pt}}
	\tikzset{dashededge/.style = {dashed,line width=1.2pt}}
\node[agent,fill=white!80!green]   (1) at (0,3) {\footnotesize{$a_1$}};
\node[agent,fill=white!80!red]   (2) at (0,1.5){\footnotesize{$a_2$}};
\node[agent]   (3) at (0,0)   {\footnotesize{$a_3$}};
\node[agentutility]   (12) at (-0.6,3) {\footnotesize{$0$}};
\node[agentutility]   (13) at (-0.6,1.5){\footnotesize{$-5$}};
\node[agentutility]   (14) at (-0.6,0)   {\footnotesize{$-2$}};
\node[chore] (4) at (2,3) {\footnotesize{$c_1$}};
\node[chore] (5) at (2,2) {\footnotesize{$c_2$}};
\node[chore] (6) at (2,1) {\footnotesize{$c_3$}};
\node[chore] (7) at (2,0) {\footnotesize{$c_4$}};
\node[choreprice] (8) at (2.6,3) {\footnotesize{\$1}};
\node[choreprice] (9) at (2.6,2) {\footnotesize{\$2}};
\node[choreprice] (10) at (2.6,1) {\footnotesize{\$2}};
\node[choreprice] (11) at (2.6,0) {\footnotesize{\$2}};
\draw[dashededge] (1) to node [near start,fill=white,inner sep=1pt] (14) {$-1$} (4);
\draw[edge] (2) to node [near start,fill=white,inner sep=1pt] (24) {$-1$} (4);
\draw[edge] (2) to node [near start,fill=white,inner sep=1pt] (25) {$-2$} (5);
\draw[edge] (2) to node [near start,fill=white,inner sep=1pt] (26) {$-2$} (6);
\draw[edge] (3) to node [near start,fill=white,inner sep=1pt] (37) {$-2$} (7);
\end{tikzpicture}
%}%end of subfloat
\caption{}
\end{subfigure}
\hspace{0.3in}
\begin{subfigure}[b]{0.2\textwidth}
%\subfloat[][]{
\begin{tikzpicture}
	\tikzset{agent/.style = {shape=circle,draw,inner sep=1pt}}
	\tikzset{agentutility/.style = {shape=circle,inner sep=0pt}}
	\tikzset{chore/.style = {shape=circle,draw,inner sep=1pt}}
	\tikzset{choreprice/.style = {shape=circle,inner sep=0pt}}
	\tikzset{edge/.style = {solid,line width=1.2pt}}
	\tikzset{dashededge/.style = {dashed,line width=1.2pt}}
\node[agent,fill=white!80!green]   (1) at (0,3) {\footnotesize{$a_1$}};
\node[agent,fill=white!80!red]   (2) at (0,1.5){\footnotesize{$a_2$}};
\node[agent]   (3) at (0,0)   {\footnotesize{$a_3$}};
\node[agentutility]   (12) at (-0.6,3) {\footnotesize{$-1$}};
\node[agentutility]   (13) at (-0.6,1.5){\footnotesize{$-4$}};
\node[agentutility]   (14) at (-0.6,0)   {\footnotesize{$-2$}};
\node[chore] (4) at (2,3) {\footnotesize{$c_1$}};
\node[chore] (5) at (2,2) {\footnotesize{$c_2$}};
\node[chore] (6) at (2,1) {\footnotesize{$c_3$}};
\node[chore] (7) at (2,0) {\footnotesize{$c_4$}};
\node[choreprice] (8) at (2.6,3) {\footnotesize{\$1}};
\node[choreprice] (81) at (2.9,3) {\footnotesize{$\downarrow$}};
\node[choreprice] (9) at (2.6,2) {\footnotesize{\$2}};
\node[choreprice] (10) at (2.6,1) {\footnotesize{\$2}};
\node[choreprice] (11) at (2.6,0) {\footnotesize{\$2}};
\draw[edge] (1) to node [near start,fill=white,inner sep=1pt] (14) {$-1$} (4);
\draw[dashededge] (2) to node [near start,fill=white,inner sep=1pt] (24) {$-1$} (4);
\draw[edge] (2) to node [near start,fill=white,inner sep=1pt] (25) {$-2$} (5);
\draw[edge] (2) to node [near start,fill=white,inner sep=1pt] (26) {$-2$} (6);
\draw[edge] (3) to node [near start,fill=white,inner sep=1pt] (37) {$-2$} (7);
\end{tikzpicture}
%}%end of subfloat
\caption{}
\end{subfigure}
\hspace{0.3in}
\begin{subfigure}[b]{0.2\textwidth}
%\subfloat[][]{
\begin{tikzpicture}
	\tikzset{agent/.style = {shape=circle,draw,inner sep=1pt}}
	\tikzset{agentutility/.style = {shape=circle,inner sep=0pt}}
	\tikzset{chore/.style = {shape=circle,draw,inner sep=1pt}}
	\tikzset{choreprice/.style = {shape=circle,inner sep=0pt}}
	\tikzset{edge/.style = {solid,line width=1.2pt}}
	\tikzset{dashededge/.style = {dashed,line width=1.2pt}}
\node[agent,fill=white!80!green]   (1) at (0,3) {\footnotesize{$a_1$}};
\node[agent,fill=white!80!red]   (2) at (0,1.5){\footnotesize{$a_2$}};
\node[agent]   (3) at (0,0)   {\footnotesize{$a_3$}};
\node[agentutility]   (12) at (-0.6,3) {\footnotesize{$-1$}};
\node[agentutility]   (13) at (-0.6,1.5){\footnotesize{$-4$}};
\node[agentutility]   (14) at (-0.6,0)   {\footnotesize{$-2$}};
\node[chore] (4) at (2,3) {\footnotesize{$c_1$}};
\node[chore] (5) at (2,2) {\footnotesize{$c_2$}};
\node[chore] (6) at (2,1) {\footnotesize{$c_3$}};
\node[chore] (7) at (2,0) {\footnotesize{$c_4$}};
\node[choreprice] (8) at (2.6,3) {\footnotesize{\$0.4}};
\node[choreprice] (9) at (2.6,2) {\footnotesize{\$2}};
\node[choreprice] (10) at (2.6,1) {\footnotesize{\$2}};
\node[choreprice] (11) at (2.6,0) {\footnotesize{\$2}};
\draw[edge] (1) to node [near start,fill=white,inner sep=1pt] (14) {$-1$} (4);
\draw[dashededge] (1) to node [near start,fill=white,inner sep=1pt] (15) {$-5$} (5);
\draw[dashededge] (1) to node [near start,fill=white,inner sep=1pt] (16) {$-5$} (6);
\draw[dashededge] (1) to node [near start,fill=white,inner sep=1pt] (17) {$-5$} (7);
\draw[edge] (2) to node [near start,fill=white,inner sep=1pt] (25) {$-2$} (5);
\draw[edge] (2) to node [near start,fill=white,inner sep=1pt] (26) {$-2$} (6);
\draw[edge] (3) to node [near start,fill=white,inner sep=1pt] (37) {$-2$} (7);
\end{tikzpicture}
%}%end of subfloat
\caption{}
\end{subfigure}
\hspace{0.3in}
\begin{subfigure}[b]{0.2\textwidth}
%\subfloat[][]{
\begin{tikzpicture}
	\tikzset{agent/.style = {shape=circle,draw,inner sep=1pt}}
	\tikzset{agentutility/.style = {shape=circle,inner sep=0pt}}
	\tikzset{chore/.style = {shape=circle,draw,inner sep=1pt}}
	\tikzset{choreprice/.style = {shape=circle,inner sep=0pt}}
	\tikzset{edge/.style = {solid,line width=1.2pt}}
	\tikzset{dashededge/.style = {dashed,line width=1.2pt}}
\node[agent]   (1) at (0,3) {\footnotesize{$a_1$}};
\node[agent]   (2) at (0,1.5){\footnotesize{$a_2$}};
\node[agent]   (3) at (0,0)   {\footnotesize{$a_3$}};
\node[agentutility]   (12) at (-0.6,3) {\footnotesize{$-6$}};
\node[agentutility]   (13) at (-0.6,1.5){\footnotesize{$-2$}};
\node[agentutility]   (14) at (-0.6,0)   {\footnotesize{$-2$}};
\node[chore] (4) at (2,3) {\footnotesize{$c_1$}};
\node[chore] (5) at (2,2) {\footnotesize{$c_2$}};
\node[chore] (6) at (2,1) {\footnotesize{$c_3$}};
\node[chore] (7) at (2,0) {\footnotesize{$c_4$}};
\node[choreprice] (8) at (2.6,3) {\footnotesize{\$0.4}};
\node[choreprice] (9) at (2.6,2) {\footnotesize{\$2}};
\node[choreprice] (10) at (2.6,1) {\footnotesize{\$2}};
\node[choreprice] (11) at (2.6,0) {\footnotesize{\$2}};
\draw[edge] (1) to node [near start,fill=white,inner sep=1pt] (14) {$-1$} (4);
\draw[edge] (1) to node [near start,fill=white,inner sep=1pt] (15) {$-5$} (5);
\draw[edge] (2) to node [near start,fill=white,inner sep=1pt] (26) {$-2$} (6);
\draw[edge] (3) to node [near start,fill=white,inner sep=1pt] (37) {$-2$} (7);
\end{tikzpicture}
%}%end of subfloat
\caption{}
\end{subfigure}
\caption{Executing the \EQ{1}+\PO{} algorithm from \Cref{thm:EQ1+PO_pseudopolynomial} on the instance in \Cref{eg:Leximin_fails_EQl_and_EF1}. The solid (respectively, dashed) edges denote items that are allocated to (respectively, in the \MBB{} set of) an agent. The edge labels denote the valuations, and the numbers next to the agent and chore nodes denote the utilities and the prices, respectively.}
\label{fig:Market_Example}
\end{figure}
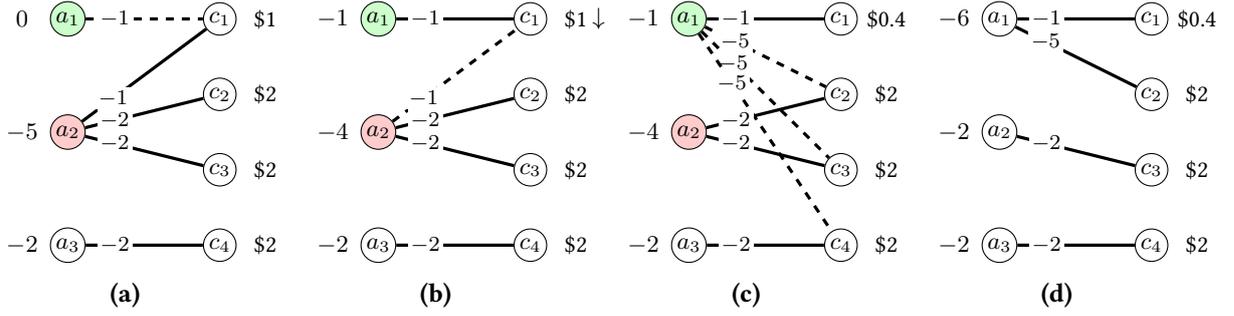

\begin{example}
Consider once again the instance in \Cref{eg:Leximin_fails_EQl_and_EF1}. Our algorithm in \Cref{thm:EQ1+PO_pseudopolynomial} works in three phases. In Phase 1, the algorithm creates an equilibrium allocation by assigning each chore to an agent that has the highest valuation for it and setting its price to be (the absolute value of) the owner's valuation; see~\Cref{fig:Market_Example}a. This ensures that the allocation satisfies the \emph{maximum bang-per-buck} or \MBB{} property (i.e., each agent's bundle consists only of items with the highest valuation-to-price ratio for that agent). The \MBB{} property guarantees that the allocation at hand is an equilibrium of some Fisher market, and therefore Pareto optimal.

The allocation constructed in Phase 1 is not \EQ{1} as $a_2$ gets three negatively valued chores and $a_1$ gets none. So, the algorithm switches to Phase 2, where it uses local search to address the equitability violations. Specifically, if there is an \EQ{1} violation, then there must be one involving the `happiest' agent, i.e., agent with the highest utility (shaded in green in~\Cref{fig:Market_Example}a). The algorithm now proceeds to transferring the chores, one at a time, from unhappy agents to the happiest agent while ensuring that all exchanges take place in an \MBB{}-consistent manner. In our example, the chore $c_1$, which is already in the \MBB{} set of agent $a_1$, is transferred from $a_2$ to $a_1$ (see~\Cref{fig:Market_Example}b).

Despite the aforementioned exchange, the allocation is still not \EQ{1} as $\{a_1,a_2\}$ once again constitute a violating pair. Furthermore, the happiest agent is already assigned its unique \MBB{} chore, so no additional \MBB{}-consistent transfers are possible. Thus, the algorithm switches to Phase 3.

In Phase 3, the algorithm \emph{creates} new \MBB{} edges in the agent-item graph by changing the prices. Specifically, the price of chore $c_1$ is \emph{lowered} until one or more of the remaining chores enter the \MBB{} set of agent $a_1$. Indeed, once the price of $c_1$ is lowered from $\$1$ to $\$0.4$, all other chores become \MBB{} for agent $a_1$ (see~\Cref{fig:Market_Example}c). As soon as the opportunity for \MBB{}-consistent exchange becomes available, the algorithm switches back to Phase 2 to perform an exchange. This time, chore $c_2$ is transferred from $a_2$ to $a_1$ (see~\Cref{fig:Market_Example}d). The new allocation is \EQ{1}, so the algorithm terminates and returns the current allocation as output.
\end{example}

\begin{remark}
We already know from \Cref{eg:Nonexistence_EQx+PO_Two_Agents} that \EQX{}+\PO{} is a strictly more demanding property combination than \EQ{1}+\PO{} in terms of \emph{existence}. That is, an \EQX{}+\PO{} allocation might fail to exist even though an \EQ{1}+\PO{} allocation is guaranteed to exist (\Cref{thm:EQ1+PO_pseudopolynomial}). Our results in \Cref{thm:EQX_PO_Strong_NP-hardness,thm:EQ1+PO_pseudopolynomial} show a similar separation between the two notions in terms of \emph{computation}: Although an \EQ{1}+\PO{} allocation can be computed in pseudopolynomial-time (\Cref{thm:EQ1+PO_pseudopolynomial}), there cannot be a pseudopolynomial-time algorithm for checking the existence of \EQX{}+\PO{} allocations unless P=NP. 
\end{remark}

\subsection{Equitability, Pareto Optimality, and Envy-Freeness}
\label{subsec:EQ_PO_EF}

We will now consider all three notions---equitability, envy-freeness, and Pareto optimality---together. It turns out that the existence result for \EQ{1}+\PO{} allocations does not hold up when we also require \EF{1} (\Cref{prop:Nonexistence_EQ1+EF1+PO}).

\begin{restatable}[\textbf{Non-existence of \EQ{1}+\EF{1}+\PO{}}]{prop}{NonexistenceEQoneEFonePO}
 \label{prop:Nonexistence_EQ1+EF1+PO}
 There exists an instance with normalized and strictly negative valuations in which no allocation is simultaneously equitable up to one chore $(\EQ{1})$, envy-free up to one chore $(\EF{1})$, and Pareto optimal $(\PO{})$.
\end{restatable}
\begin{proof}
Consider the following instance with eight chores and four agents with normalized and strictly negative valuations:
\begin{table}[ht]
\centering
%\small
\begin{tabular}{c|cccccccc}
& $c_1$ & $c_2$ & $c_3$ & $c_4$ & $c_5$ & $c_6$ & $c_7$ & $c_8$\\ \hline
  $a_1$ & $-10$ & $-10$ & $-10$ & $-10$ & $-10$ & $-10$ & $-10$ & $-10$\\
  $a_2$ & $-10$ & $-10$ & $-10$ & $-10$ & $-10$ & $-10$ & $-10$ & $-10$\\
  $a_3$ & $-73$ & $-1$ & $-1$ & $-1$ & $-1$ & $-1$ & $-1$ & $-1$\\
  $a_4$ & $-73$ & $-1$ & $-1$ & $-1$ & $-1$ & $-1$ & $-1$ & $-1$\\
\end{tabular}
\end{table}

Suppose, for contradiction, that there exists an allocation $A$ that is \EQ{1}, \EF{1}, and \PO{}. Then, we claim that $a_1$ gets exactly one chore in $A$. Indeed, $a_1$ cannot get three or more chores in $A$, since that would result in some other agent getting at most one chore, creating an \EF{1} violation with respect to $a_1$. If $a_1$ gets exactly two chores, then either $a_3$ or $a_4$ will create an \EQ{1} violation with respect to $a_1$. This is because one of $a_3$ and $a_4$ will necessarily miss out on $c_1$ and therefore have a utility of at least $-7$ from the remaining chores. Finally, if $a_1$ does not get any chore, then one of the other agents will get at least three chores. Because of strictly negative valuations, this will create an \EQ{1} violation with $a_1$. Therefore, $a_1$ gets exactly one chore in $A$. By a similar argument, so does $a_2$.

Therefore, a total of six chores are assigned between $a_3$ and $a_4$. Assume, without loss of generality, that $a_3$ gets at least three chores. Then, whoever of $a_1$ or $a_2$ misses out on $c_1$ will create an \EF{1} violation with respect to $a_3$, giving us the desired contradiction.
\end{proof}

Turning to the computational question, we notice that the allocation constructed in the proof of \Cref{thm:EQX_PO_Strong_NP-hardness} is envy-free. Therefore, checking the existence of an \EQX{}+\PO{}+\EF{}/\EFX{}/\EF{1} allocation is also strongly \NPH{}. We note that the analogous problem in the goods setting is also known to be computationally hard~\citep{FSV+19equitable}.

\begin{restatable}[\textbf{Hardness of \EQX{}+\PO{}+\EF{}/\EFX{}/\EF{1}}]{corollary}{EQXPOEFStrongNPhardness}
 \label{cor:EQX_PO_EF/EFX/EF1_Strong_NP-hardness}
 Determining whether a given fair division instance admits an allocation that is simultaneously $X+Y+\PO{}$, where $X$ refers to equitable up to any chore $(\EQX{})$, and $Y$ refers to either envy-free $(\EF{})$, envy-free up to any chore $(\EFX{})$, or envy-free up to one chore $(\EF{1})$, is \textup{strongly \NPH{}}, even for normalized valuations.
\end{restatable}

\subsection{Equitability up to a Duplicated Chore}
\label{subsec:Dup_EQ}

In this section, we will explore a slightly different version of approximate equitability for chores wherein instead of removing a chore from the less-happy agent's bundle, we imagine \emph{adding} a chore to the happier agent's bundle. In particular, we will ask that pairwise jealousy should be removed by duplicating a single chore from the less happy agent's bundle and adding it to the happier agent's bundle.

Formally, an allocation $A$ is \emph{equitable up to one duplicated chore} (\DEQ{1}) if for every pair of agents $i,k \in [n]$ such that $A_i \neq \emptyset$, there exists a chore $j \in A_i$ such that $v_i(A_i) \geq v_k(A_k \cup \{j\})$. An allocation $A$ is \emph{equitable up to any duplicated chore} (\DEQX{}) if for every pair of agents $i,k \in [n]$ such that $A_i \neq \emptyset$ and for every chore $j \in A_i$ such that $v_{i,j} < 0$, we have $v_i(A_i) \geq v_k(A_k \cup \{j\})$.

\begin{restatable}[\textbf{Existence of \DEQX{}+\PO{}}]{prop}{DEQXPO}
 \label{prop:DEQX+PO}
Given any fair division instance with additive valuations, an allocation that is equitable up to any duplicated chore $(\DEQX{})$ and Pareto optimal $(\PO{})$ always exists.
\end{restatable}
\begin{proof}(Sketch)
We will show that any \Leximin{}-optimal allocation, say $A$, satisfies \DEQX{}~(Pareto optimality is easy to verify). Suppose, for contradiction, that there exist agents $i,k \in [n]$ with $A_i \neq \emptyset$ and some chore $j \in A_i$ such that $v_{i,j}<0$ and $v_i(A_i) < v_k(A_k \cup \{j\})$. Let $B$ be an allocation derived from $A$ by transferring the chore $j$ from agent $i$ to agent $k$. That is, $B_i \coloneqq A_i \setminus \{j\}$, $B_k \coloneqq A_k \cup \{j\}$ and $B_h \coloneqq A_h$ for all $h \in [n] \setminus \{i,k\}$. Since \DEQX{} is violated with respect to chore $j$, we have that $v_{i,j} < 0$, and therefore $v_i(B_i) = v_i(A_i) - v_{i,j} > v_i(A_i)$. Furthermore, $v_k(B_k) = v_k(A_k \cup \{j\}) > v_i(A_i)$ by the \DEQX{} violation condition. The utility of any other agent is unchanged. Therefore, $B$ is a `Leximin improvement' over $A$, which is a contradiction.
\end{proof}

Thus, \Cref{prop:DEQX+PO} shows that the duplicate version of approximate equitability (\DEQX{}) compares favorably against the standard version (\EQX{}) in the sense that a \DEQX{}+\PO{} allocation is guaranteed to exist whereas an \EQX{}+\PO{} allocation might not exist even with two agents and strictly negative valuations (\Cref{eg:Nonexistence_EQx+PO_Two_Agents}).

On the computational side, we find that a \DEQ{1} allocation of chores can be computed in polynomial time via a greedy algorithm.

\begin{restatable}[]{prop}{dEQoneComputation}
 \label{prop:dEQ1_Computation}
 A \DEQ{1} allocation of chores always exists and can be computed in polynomial time.
\end{restatable}
\begin{proof}
Consider any fixed ordering $j_1,j_2,\dots,j_m$ of the chores. Our algorithm assigns the chore $j_t$ in the $t^\text{th}$ round. Let $A^{t-1}$ denote the (partial) allocation at the end of $t-1$ rounds. The algorithm assigns the chore $j_t$ to the agent $i_t$ defined as follows:
\begin{equation*}
\displaystyle i_t \in \arg\max_{i \in [n]} v_i(A^{t-1}_i \cup \{j_t\}).
%\label{eqn:DEQ1_assignment_rule}
\end{equation*}
That is, in a thought experiment where each agent gets a copy of the chore $j_t$, agent $i_t$ has the highest utility in the derived allocation. It is easy to see that the algorithm runs in polynomial time.

We will now use induction to show that the algorithm maintains a \DEQ{1} (partial) allocation at every step. This is certainly true prior to the first round, since an empty allocation is \DEQ{1}. Suppose the partial allocations at the end of each of the first $t-1$ rounds, namely $A^1,\dots,A^{t-1}$, satisfy \DEQ{1}. We will argue that the same is true for the (partial) allocation $A^t$ at the end of the $t^\text{th}$ round.

Suppose, for contradiction, that $A^t$ fails \DEQ{1}. That is, there exists a pair of agents $i,k \in [n]$ with $A^t_i \neq \emptyset$ such that $v_i(A^t_i) < v_k(A^t_k \cup \{j\})$ for every chore $j \in A^t_i$. Then, the chore $j_t$  must have been assigned to agent $i$, i.e., $j_t \in A^t_i$. Indeed, if $j_t$ were to be assigned to any agent in $[n] \setminus \{i,k\}$, then the \DEQ{1} violation between $i$ and $k$ would have existed during round $t-1$, contradicting the fact that $A^{t-1}$ satisfies \DEQ{1}. Furthermore, if $j_t$ were to be assigned to agent $k$, then agent $k$'s utility in round $t-1$ would have strictly exceeded its utility in round $t$, implying once again that \DEQ{1} violation between $i$ and $k$ would have existed in round $t-1$, which is a contradiction. Therefore, the chore $j_t$ must have been assigned to agent $i$ in round $t$.

We can now instantiate the \DEQ{1} violation condition for the chore $j_t$ to get $v_i(A^t_i) < v_k(A^t_k \cup \{j_t\})$. Note that since $j_t$ is assigned to agent $i$, the bundle of agent $k$ remains unchanged between rounds $t-1$ and $t$, and therefore $A^t_k = A^{t-1}_k$ and $A^t_i = A^{t-1}_i \cup \{j_t\}$. Therefore, the \DEQ{1} violation can be rewritten as $v_i(A^{t-1}_i \cup \{j_t\}) < v_k(A^{t-1}_k \cup \{j_t\})$. This implies that $i$ is \emph{not} the highest utility agent in the thought experiment where each agent is assigned a (hypothetical) copy of the chore $j_t$, which is a contradiction. Therefore, the allocation $A^t$ must satisfy \DEQ{1}. By induction, the same holds for the allocation returned by the algorithm.
\end{proof}

Unfortunately, the greedy algorithm in \Cref{prop:dEQ1_Computation} does not guarantee a \DEQ{X} allocation. This stands in contrast to the situation for \EQ{X}, which is easily achieved by a greedy procedure. Settling the complexity of computing \DEQ{X} allocations is an interesting question for future work.

The complexity of computing an allocation that satisfies either \DEQ{1}+\PO{} or the stronger \DEQ{X}+\PO{} also remains open. For \DEQ{1}+\PO{}, a natural approach would be to apply the market techniques used in \Cref{thm:EQ1+PO_pseudopolynomial}, but that would require care as \DEQ{1} lacks the following ``monotonicity'' property that \EQ{1} has: If an allocation is not \EQ{1}, then without loss of generality, there exists a violation with respect to the happiest agent. The same is not true for violations of \DEQ{1}, which makes the analysis less obvious.

In \Cref{subsec:Zero-Valued-Chores-Removal-Variants}, we explore a variant of \DEQ{X}, denoted as \DEQXzero{}, in which the $v_{i,j}<0$ condition is not imposed on the duplicated chore $j$. With this modification, we show that computing an allocation satisfying \DEQXzero{}+\PO{} is \NPhard{}, as well as an equivalent result for the analogous notion of \EQXzero{}.

\begin{remark}[\textbf{A tractable special case: binary valuations}]
An instance is said to have \emph{binary valuations} if for every agent $i \in [n]$ and every chore $j \in [m]$, we have $v_{i,j} \in \{-1,0\}$. For this restricted setting, there is a simple polynomial-time algorithm that gives an \EQX{}+\DEQX{}+\EFX{}+\PO{} allocation, as follows: If a chore is valued at $0$ by one or more agents, then it is arbitrarily assigned to an agent that values it at $0$. The remaining chores, which are valued at $-1$ by every agent, are assigned in a round-robin fashion.
\end{remark}
\section{Experiments}
\label{sec:Experiments}
In this section, we will compare various algorithms in terms of how frequently they satisfy different combinations of fairness and efficiency properties on synthetic as well as real-world datasets.

%% Datasets
For \emph{synthetic} data, we follow the setup of \citet{FSV+19equitable} for goods by fixing $n=5$ agents, $m=20$ chores, and generating $1000$ instances with (the negation of) the valuations drawn from Dirichlet distribution. Additional pre-processing is required to ensure that the valuations are integral and normalized (see \Cref{subsec:Experiments_Appendix}). Recall that integral valuations are required for \Cref{thm:EQ1+PO_pseudopolynomial}. None of our results require normalization, but it is a natural condition to impose in practice.

The \emph{real-world} dataset consists of $2613$ instances obtained from the \emph{Spliddit} website~\citep{GP15spliddit}, with the number of agents ranging from $2$ to $15$, and the number of distinct chores ranging from $3$ to $1100$. Unlike the goods case, the ``task division'' segment of Spliddit allows distinct items to have multiple \emph{copies}.\footnote{\url{http://www.spliddit.org/apps/tasks}} Furthermore, instead of directly eliciting additive valuations (as is the case for goods), the website asks the users to specify their preferences in the form of \emph{multipliers}; that is, given two chores $c_1$ and $c_2$, how many times would a user be willing to complete $c_1$ instead of completing $c_2$ once.\footnote{For example, doing laundry $2.5$ times could be equivalent to washing dishes once.} As a result, the elicited valuations might not be integral. These design features force us to make a number of pre-processing decisions (see \Cref{subsec:Experiments_Appendix}). In particular, in order to ensure integrality of valuations and remain as faithful as possible to the Spliddit instances, we have to give up on normalization.

%% Algorithms
We consider the following four algorithms: (1) The \emph{greedy} algorithm from \Cref{prop:EQX_Existence_Computation}, (2) the \emph{Leximin} solution, (3) the market-based algorithm \AlgEQonePO{} from \Cref{thm:EQ1+PO_pseudopolynomial}, and (4) an algorithm currently deployed on the \emph{Spliddit} website for dividing chores. The latter is a randomized algorithm that computes an ex ante equitable lottery over integral allocations (refer to \Cref{subsec:Experiments_Appendix} for details).

\begin{figure}
	\centering
	\includegraphics[width=\linewidth]{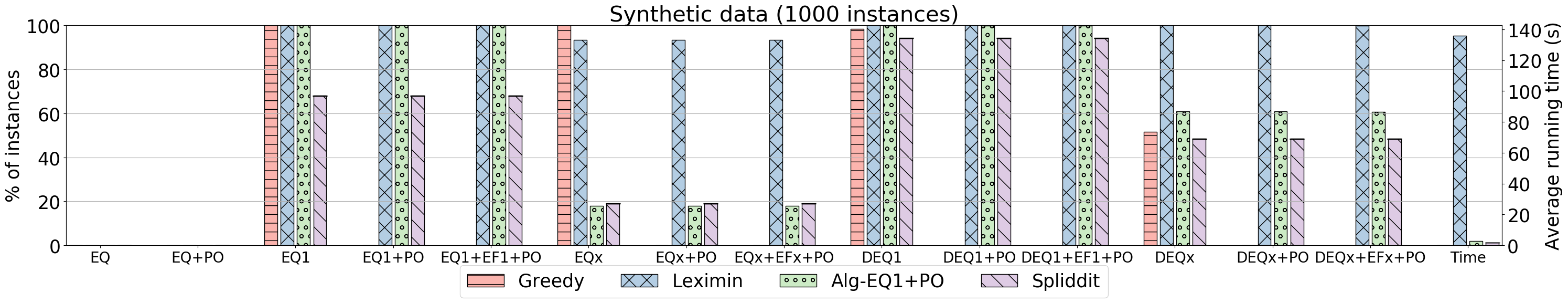}
	\includegraphics[width=\linewidth]{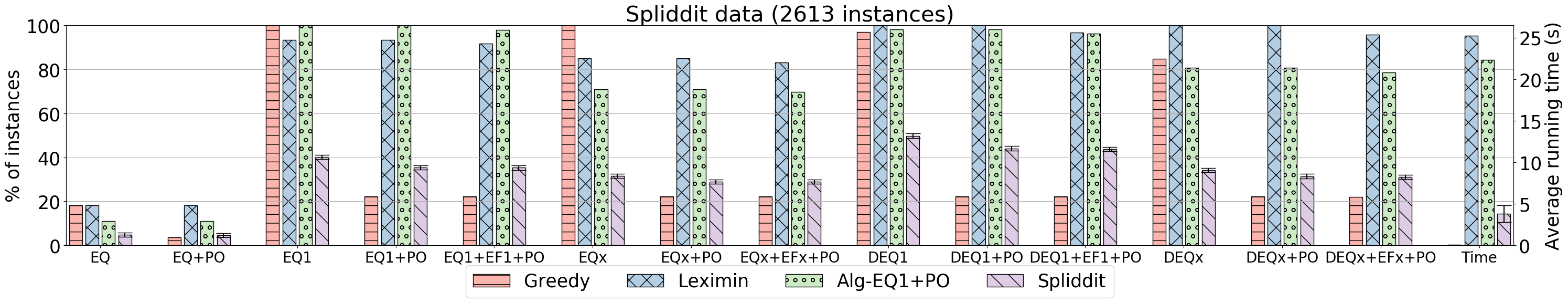}
	\caption{Experimental results for Synthetic (top) and  Spliddit (bottom) datasets.}
	\label{fig:Experimental_Results}
\end{figure}

%% Results
\Cref{fig:Experimental_Results} presents our experimental results. For each property combination (X-axis), the plots show the $\%$ of instances (Y-axis) for which each algorithm achieves those properties. The rightmost set of bars present a comparison of the running times. For the Spliddit algorithm, we plot the average values obtained from 100 runs, and the error bars show one standard deviation around the mean.

Starting with exact equitability, we observe that a very small fraction of instances ($<20\%$ in Spliddit and none in Synthetic) admit \EQ{} and \EQ{}+\PO{} allocations, as one might expect.\footnote{An equitable (\EQ{}) and Pareto optimal allocation (\PO{}), whenever it exists, is provably achieved by the Leximin algorithm.} For the approximate notions, the greedy algorithm finds \EQX{} allocations on all instances as advertised (\Cref{prop:EQX_Existence_Computation}), but its performance drops off sharply when \PO{} is also required; in particular, for Synthetic data, the greedy outcome is always Pareto dominated.

Leximin performs remarkably well across the board. In addition to satisfying \DEQX{}+\PO{} on all instances (\Cref{prop:DEQX+PO}), it also satisfies \EQ{X} and \EF{X} on more than $80\%$ of the instances in both datasets. Unfortunately, it is also the slowest of all algorithms, with an average runtime of $\sim$$140$ seconds on Synthetic dataset, compared to $<$$1$ second runtime of the fastest (greedy) algorithm.

The market-based algorithm \AlgEQonePO{} computes \EQ{1}+\PO{} allocations as expected (\Cref{thm:EQ1+PO_pseudopolynomial}), and somewhat surprisingly, also satisfies \DEQ{1} (and \EF{1}). However, its performance drops off when stronger approximations of \EQX{}/\DEQX{} are required.

The Spliddit algorithm is consistently (and often, significantly) outperformed by Leximin and \AlgEQonePO{}, even on the Spliddit dataset. The reason is that the Spliddit algorithm is perfectly equitable \emph{ex ante} but not necessarily \EQ{1} \emph{ex post}. As a result, it is better suited for ensuring fairness over time, say, when the same set of chores are repeatedly divided among the same agents, as noted on the Spliddit website.

%% Conclusions
In summary, Leximin emerges as the algorithm of choice in terms of simultaneously achieving approximate fairness and economic efficiency. We find it intriguing that the same algorithm was also a clear winner in the experimental analysis of \citet{FSV+19equitable} for goods, even though it is no longer provably \EQX{} (or even \EQ{1}). Equally intriguing is the fact that a currently deployed algorithm is outperformed by well-known (Leximin) and proposed (\AlgEQonePO{)} algorithms, thereby justifying the usefulness of analyzing (approximate) fairness for chore division.
\section{Discussion}
\label{sec:Discussion}
We studied equitable allocations of indivisible chores in conjunction with other well-known notions of fairness (envy-freeness) and economic efficiency (Pareto optimality), and provided a number of existential and computational results. Our results reveal some interesting points of difference between the goods and chores settings. While a modification of the market approach used by~\citet{FSV+19equitable} to achieve \EQ{1}+\PO{} in the goods setting works for chores, it may be the case that no allocation satisfying \EQ{X}+\PO{} exists in the chores setting. In response to this possible nonexistence, we have defined two new notions of relaxed equitability, \DEQ{1} and \DEQ{X}, that address equitability violations by \emph{adding} chores to bundles rather than removing them. A number of open questions remain regarding the computation of allocations that satisfy these notions (with or without Pareto optimality). It may also be an interesting topic for future work to consider similar relaxations of envy-freeness in the chores setting.

In our experimental analysis, we have considered four different algorithms for chore division on both a real-world dataset gathered from the Spliddit website as well as a synthetic dataset.
Our experiments present a compelling case that, in practice, Leximin is the best known algorithm for one-shot allocation of indivisible chores. This is true not only with respect to (relaxed) equitability, but also (relaxed) envy-freeness and Pareto optimality.

\section*{Acknowledgments}
We are grateful to the anonymous reviewers for their helpful comments, and to Ariel Procaccia and Nisarg Shah for sharing with us the data from Spliddit. LX acknowledges NSF \#1453542 and \#1716333 for support.

\bibliographystyle{named}%{alpha}
\bibliography{References}

\clearpage
\newpage
\section{Appendix}
 \label{sec:Appendix}

\subsection{Proof of Theorem~\ref{thm:EQ1+PO_pseudopolynomial}}
%\subsection{A Pseudopolynomial-time Algorithm for \EQ{1} + \PO{}}
\label{subsec:EQ1+PO}

Recall the statement of \Cref{thm:EQ1+PO_pseudopolynomial}.

\EQonePOPseudopolynomial*

\begin{remark}
Given an instance $\I$ and a chore $j \in [m]$, let $S_j \coloneqq \{i \in [n] : v_{i,j} = 0\}$ denote the set of agents that value $j$ at $0$. Then, in any Pareto optimal allocation, $j$ must be assigned to one of the agents in $S_j$. The choice of which agent in $S_j$ gets $j$ is immaterial from the viewpoint of \EQ{1}. (Specifically, if $A$ is an \EQ{1}+\PO{} allocation that assigns chore $j$ to some agent $i \in S_j$, then an allocation derived from $A$ in which chore $j$ is assigned to some other agent $k \in S_j$ is also \EQ{1}+\PO{}.) Therefore, in our discussion on \EQ{1}+\PO{} allocations, we will only focus on strictly negative valuations.
\label{rem:EQ1_PO_Negative_Vals}
\end{remark}

The proof of \Cref{thm:EQ1+PO_pseudopolynomial} relies on the algorithm \AlgEQonePO{} (presented in Algorithm~\ref{alg:EQ1+PO}), and spans \Cref{subsec:EQ1+PO,subsec:Proof_RunningTime_ALG_EQ1+PO,subsec:Proof_RunningTime_Phase2,subsec:Proof_RunningTime_Phase3,subsec:Proof_Correctness_ALG_EQ1+PO_original_instance}.
We will start with some necessary definitions that will help us state \Cref{thm:epsEQ1+PO_pseudopolynomial}, of which \Cref{thm:EQ1+PO_pseudopolynomial} is a special case.

\paragraph{Fractional allocations}
A \emph{fractional allocation} $\x \in [0,1]^{n \times m}$ refers to a fractional assignment of the chores to the agents such that exactly one unit of any chore is allocated, i.e., for every chore $j \in [m]$, $\sum_{i \in [n]} x_{i,j} = 1$. We will use the term \emph{allocation} to refer to an integral (or discrete) allocation and explicitly write \emph{fractional allocation} otherwise.

\paragraph{$\eps$-Pareto optimality}
Given any $\eps \geq 0$, $A$ is \emph{$\eps$-Pareto optimal} ($\eps$-\PO{}) if there does not exist an allocation $B$ such that $v_k(B_k) \geq \frac{1}{(1+\eps)} v_k(A_k)$ for every agent $k \in [n]$ with one of the inequalities being strict. 

\paragraph{Fractional Pareto optimality}
An allocation is \emph{fractionally Pareto optimal} (\fPO{}) if it not Pareto dominated by any fractional allocation. Thus, a fractionally Pareto optimal allocation is also Pareto optimal, but the converse is not necessarily true.% (\Cref{prop:Nonexistence_EQx+fPO}).

\paragraph{$\eps$-\EQ{1} allocation}
Given any $\eps \geq 0$, an allocation $A$ is $\eps$-\emph{equitable up to one chore} ($\eps$-\EQ1) if for every pair of agents $i,k \in [n]$ such that $A_i \neq \emptyset$, there exists a chore $j \in A_i$ such that $\frac{1}{(1+\eps)} v_i(A_i \setminus \{j\}) \geq v_k(A_k)$.

\begin{restatable}{theorem}{epsEQonePOPseudopolynomial}
 \label{thm:epsEQ1+PO_pseudopolynomial}
 Given any fair division instance with additive and strictly negative valuations and any $\eps > 0$, an allocation that is $3\eps$-equitable up to one chore $(3\eps\text{-}\EQ{1})$ and $\eps$-Pareto optimal $(\eps\text{-}\PO{})$ always exists and can be computed in $\O(\poly(m,n,\ln |v_{\min}|,\nicefrac{1}{\eps}))$ time, where $v_{\min} = \min_{i,j} v_{i,j}$.
\end{restatable}

When $0 < \eps \leq \frac{1}{6m |v_{\min}|^3}$, we recover \Cref{thm:EQ1+PO_pseudopolynomial} as a special case of \Cref{thm:epsEQ1+PO_pseudopolynomial} (see \Cref{lem:Bound_On_Eps_For_Exact_PO,lem:Bound_On_eps_for_Exact_EQone}). 

The remainder of this section develops the necessary preliminaries that will enable us to present our algorithm (Algorithm~\ref{alg:EQ1+PO}) and the analysis of its running time (\Cref{lem:RunningTime_ALG_EQ1+PO}) and correctness (\Cref{lem:Correctness_ALG_EQ1+PO_original_instance}). The detailed proofs of these results are presented subsequently in \Cref{subsec:Proof_RunningTime_ALG_EQ1+PO,subsec:Proof_RunningTime_Phase2,subsec:Proof_RunningTime_Phase3,subsec:Proof_Correctness_ALG_EQ1+PO_original_instance}.

\subsubsection*{Market Preliminaries}
%\label{subsec:Market_Preliminaries}

\paragraph{Fisher market for chores}
 A Fisher market for chores is an economic model that consists of a set of divisible chores and a set of agents (or buyers), each of whom is given a \emph{budget} (or \emph{endowment}) of virtual money \citep{BS00compute}. The agents are required to exhaust their budgets (of virtual money) to purchase a utility-maximizing subset of the chores but do not derive any utility from the money itself. Formally, a Fisher market is given by a tuple $\M = \langle [n],[m],\V,e \rangle$ consisting of a set of $n$ \emph{agents} $[n] = \{1,2,\dots,n\}$, a set of $m$ divisible \emph{chores} $[m] = \{1,2,\dots,m\}$, a \emph{valuation profile} $\V = \{v_1,v_2,\dots,v_n\}$ and a vector of \emph{endowments} or \emph{budgets} $\e = (e_1,e_2,\dots,e_n)$. 

A \emph{market outcome} refers to a pair $(A,\p)$, where $A = (A_1,\dots,A_n)$ is a \emph{fractional allocation} of the $m$ chores, and $\p = (p_1,\dots,p_m)$ is a \emph{price vector} that associates a non-negative price $p_j \geq 0$ with every chore $j \in [m]$. The \emph{spending} of agent $i$ under the market outcome $(A,\p)$ is given by $s_i = \sum_{j = 1}^m A_{i,j} p_j$. The \emph{utility} derived by the agent $i$ under $(A,\p)$ depends linearly on the valuations as $v_i(A_i) =  \sum_{j=1}^m A_{i,j} v_{i,j}$.

\paragraph{\MBB{} ratio and \MBB{} set}
Given a price vector $\p = (p_1,\dots,p_m)$, define the \emph{bang-per-buck} ratio of agent $i$ for chore $j$ as $\alpha_{i,j} \coloneqq v_{i,j}/p_j$. The \emph{maximum bang-per-buck} \emph{ratio} (or \MBB{} ratio) of agent $i$ is $\alpha_i \coloneqq \max_j \alpha_{i,j}$.\footnote{If $v_{i,j} = 0$ and $p_j = 0$, then $\alpha_{i,j} \coloneqq 0$.} The \emph{maximum bang-per-buck} \emph{set} (or \MBB{} set) of agent $i$ is the set of all chores that maximize the bang-per-buck ratio for agent $i$ at the price vector $\p$, i.e., $\MBB_i \coloneqq \{j \in [m] : v_{i,j}/p_j = \alpha_i\}$. Note that the \MBB{} ratios are non-positive.

A market outcome $(A,\p)$ constitutes an \emph{equilibrium} if it satisfies the following conditions:
	\begin{itemize}
	\item \emph{Market clearing}: Each chore is either priced at zero or is completely allocated. That is, for every chore $j \in [m]$, either $p_j = 0$ or $\sum_{i=1}^n A_{i,j} = 1$.
	\item \emph{Budget exhaustion}: Agents spend their budgets completely, i.e., $s_i = e_i$ for all $i \in [n]$.
	\item \emph{\MBB{} consistency}: Each agent's allocation is a subset of its $\MBB{}$ set. That is, for every agent $i \in [n]$ and every chore $j \in [m]$, $A_{i,j} > 0 \implies j \in \MBB_i$. Note that \MBB{} consistency implies that every agent maximizes its utility at the given prices $\p$ under the budget constraints.
	\end{itemize}

\Cref{prop:FirstWelfareTheorem} presents the well-known \emph{first welfare theorem} for Fisher markets \citep[Chapter~16]{MWG+95microeconomic}. For completeness, we provide a proof of this result for the chores setting.

\begin{restatable}[First welfare theorem]{prop}{FirstWelfareTheorem}
%\begin{proposition}
 \label{prop:FirstWelfareTheorem}
 For a Fisher market with linear utilities, any equilibrium outcome is fractionally Pareto optimal $(\fPO{})$.
%\end{proposition}
\end{restatable}

\begin{proof}
Suppose, for contradiction, that there exists an allocation $A$ and a price vector $\p$ such that $(A,\p)$ is an equilibrium but $A$ is not \fPO{}. Thus, there exists a fractional allocation, say $\x$, such that $v_i(\x_i) \geq v_i(A_i)$ for all $i \in [n]$ and $v_k(\x_k) > v_k(A_k)$ for some $k \in [n]$. Since both $\x$ and $A$ are required to assign all chores, we have that $\cup_{i \in [n]} \x_i = \cup_{i \in [n]} A_i = [m]$.

By \MBB{}-consistency, we have that $v_i(A_i) = p(A_i) \cdot \alpha_i$ for every $i \in [n]$, where $\alpha_i$ is the \MBB{} ratio for agent $i$, and $p(A_i) \coloneqq \sum_{j \in A_i} p_j$ is the price of the bundle $A_i$. Since $\x$ is not guaranteed to satisfy \MBB{}-consistency, we have that $v_i(\x_i) \leq p(\x_i) \cdot \alpha_i$ for every $i \in [n]$. Substituting these relations in the aforementioned inequalities, we get that $p(\x_i) \cdot \alpha_i \geq p(A_i) \cdot \alpha_i$ for all $i \in [n]$ and $p(\x_k) \cdot \alpha_k > p(A_k) \cdot \alpha_k$ for some $k \in [n]$.

Recall from \Cref{sec:Preliminaries} that for each chore, there exists some agent with a non-zero valuation for
it, and for each agent, there exists a chore that it has non-zero value for. This implies that $\alpha_i < 0$ for every agent $i \in [n]$. Thus, $p(\x_i) \leq p(A_i)$ for all $i \in [n]$ and $p(\x_k) < p(A_k)$ for some $k \in [n]$. By summing these inequalities for all agents, we get that $p([m]) = \sum_{i \in [n]} p(\x_i) < \sum_{i \in [n]} p(A_i) = p([m])$, which is a contradiction. Hence, $A$ must be \fPO{}. 
\end{proof}

\paragraph{\MBB{}-Allocation graph and alternating paths} 

Given a Fisher market $\M = \langle [n],[m],\V,e \rangle$, let $A$ and $\p$ denote an integral allocation and a price vector for $\M$ respectively. An \emph{\MBB{}-allocation graph} is an undirected bipartite graph $G$ with vertex set $[n] \cup [m]$ and an edge between agent $i \in [n]$ and chore $j \in [m]$ if either $j \in A_i$ (called an \emph{allocation edge}) or $j \in \MBB_i$ (called an \emph{$\MBB$ edge}). Notice that if $A$ is \MBB{}-consistent (i.e., $j \in A_i \implies j \in \MBB_i$), then the allocation edges are a subset of \MBB{} edges.

For an \MBB{}-allocation graph, define an \emph{alternating path} $P = (i,j_1,i_1,j_2,i_2,\dots,i_{\ell - 1},j_\ell,k)$ from agent $i$ to agent $k$ (and involving the agents $i_1,i_2,\dots,i_{\ell-1}$ and the chores $j_1,j_2,\dots,j_\ell$) as a series of alternating \MBB{} and allocation edges such that $j_1 \in \MBB_i \cap A_{i_1}$, $j_2 \in \MBB_{i_1} \cap A_{i_2}$,$\dots$, $j_\ell \in \MBB_{i_{\ell - 1}} \cap A_k$. If such a path exists, we say that agent $k$ is \emph{reachable} from agent $i$ via an alternating path.\footnote{Note that no agent or chore can repeat in an alternating path.} In this case, the \emph{length} of path $P$ is $2\ell$ since it consists of $\ell$ allocation edges and $\ell$ \MBB{} edges.

\paragraph{Reachability set} Let $G$ denote the \MBB{}-allocation graph of a Fisher market for the outcome $(A,\p)$. Fix a \emph{source} agent $i \in [n]$ in $G$. Define the \emph{level} of an agent $k \in [n]$ as half the length of the shortest alternating path from $i$ to $k$ if one exists (i.e., if $k$ is reachable from $i$), otherwise set the level of $k$ to be $n$. The level of the source agent $i$ is defined to be $0$. The \emph{reachability set} $\R_i$ of agent $i$ is defined as a level-wise collection of all agents that are reachable from $i$, i.e., $\R_i = (\R_i^{0},\R_i^{1},\R_i^{2},\dots,)$, where $\R_i^{\ell}$ denotes the set of agents that are at level $\ell$ with respect to agent $i$. Note that given an \MBB{}-allocation graph, a reachability set can be constructed in polynomial time via breadth-first search.

Given a reachability set $\R_i$, we can redefine an \emph{alternating path} as a set of alternating \MBB{} and allocation edges connecting agents at a \emph{lower level} to those at a \emph{higher level}. Formally, we will call a path $P = (i,j_1,i_1,j_2,i_2,\dots,i_{\ell - 1},j_\ell,k)$ \emph{alternating} if (1) $j_1 \in \MBB_i \cap A_{i_1}$, $j_2 \in \MBB_{i_1} \cap A_{i_2}$,$\dots$, $j_\ell \in \MBB_{i_{\ell - 1}} \cap A_k$, and (2) $\level(i) < \level(i_1) < \level(i_2) < \dots < \level(i_{\ell - 1}) < \level(k)$. Thus, an alternating path cannot have edges between agents at the same level.

\paragraph{Violators and path-violators} Given a Fisher market $\M$ and an allocation $A$, an agent $i \in [n]$ with the highest valuation among all the agents is called the \emph{reference agent}, i.e., $i \in \arg\max_{k \in [n]} v_k(A_k)$.\footnote{Ties are broken lexicographically.} An agent $k \in [n]$ is said to be a \emph{violator} if for every chore $j \in A_k$, we have that $v_k(A_k \setminus \{j\}) < v_i(A_i)$, where $i$ is the reference agent. Notice that the allocation $A$ is \EQ{1} if and only if there is no violator.

Given any $\eps > 0$, an agent $k \in [n]$ is an $\eps$-\emph{violator} if for every chore $j \in A_k$, we have $\frac{1}{(1+\eps)} v_k(A_k \setminus \{j\}) <  v_i(A_i)$. Thus, an agent can be a violator without being an $\eps$-violator. An allocation $A$ is $\eps$-\EQ{1} if and only if there is no $\eps$-violator.

A closely related notion is that of a \emph{path-violator}. Let $i$ and $\R_i$ denote the reference agent and its reachability set respectively. An agent $k \in \R_i$ is a \emph{path-violator} with respect to the alternating path $P = (i,j_1,i_1,j_2,i_2,\dots,i_{\ell - 1},j_\ell,k)$ if $v_k(A_k \setminus \{j_\ell\}) < v_i(A_i)$. Note that a path-violator (along a path $P$) need not be a violator as there might exist some chore $j \in A_k$ not on the path $P$ such that $v_k(A_k \setminus \{j\}) \geq v_i(A_i)$. Finally, given any $\eps > 0$, an agent $k \in \R_i$ is an \emph{$\eps$-path-violator} with respect to the alternating path $P = (i,j_1,i_1,\dots,j_\ell,k)$ if $\frac{1}{(1+\eps)} v_k(A_k \setminus \{j_\ell\}) <  v_i(A_i)$.

\paragraph{$\eps$-rounded instance}
Given any $\eps > 0$, an \emph{$\eps$-rounded instance} refers to a fair division instance $\langle [n], [m], \V \rangle$ in which the valuations are either zero or the negative of a non-negative integral power of $(1+\eps)$. That is, for every agent $i \in [n]$ and every chore $j \in [m]$, we have $v_{i,j} \in \{0,-(1+\eps)^t\}$ for some $t \in \N \cup \{0\}$. 

Given any instance $\I = \langle [n], [m], \V \rangle$, the \emph{$\eps$-rounded version} of $\I$ is an instance $\I' = \langle [n], [m], \W \rangle$ obtained by \emph{rounding down} the valuations in $\I$ to the nearest integral power of $(1+\eps)$. That is, the \emph{$\eps$-rounded version} of instance $\I = \langle [n], [m], \V \rangle$ is an $\eps$-rounded instance $\I' = \langle [n], [m], \W \rangle$ constructed as follows: For every agent $i \in [n]$ and every chore $j \in [m]$, $ w_{i,j}:= -(1+\eps)^{\lceil \log_{1+\eps}|v_{i,j}| \rceil}$ if $v_{i,j}< 0$, and $0$ otherwise. Notice that $v_{i,j} \geq w_{i,j} \geq (1+\eps)v_{i,j}$ for every agent $i$ and every chore $j$. We will assume that the rounded valuations are also \emph{additive}, i.e., for any set of chores $S \subseteq [m]$, $w_i(S) \coloneqq  \sum_{j \in S} w_{i,j}$.

\subsubsection*{Description of the Algorithm}

Given an instance $\I = \langle [n],[m],\V \rangle$ as input, we first construct its $\eps$-rounded version $\I' = \langle [n],[m],\W \rangle$, which is then provided as an input to \AlgEQonePO{} (Algorithm~\ref{alg:EQ1+PO}). 

The algorithm consists of three phases. In Phase 1, each chore is assigned to an agent with the highest (i.e., closest to zero) valuation for it (Line~\ref{algline:Assignment_Phase1}). This ensures that the initial allocation is \emph{integral} as well as \emph{fractionally Pareto optimal} (\fPO{}).\footnote{Indeed, the said allocation is \MBB{}-consistent with respect to the prices in Line~\ref{algline:Prices_Phase1}, and is therefore an equilibrium outcome of a Fisher market in which each agent is provided a budget equal to its spending under the allocation. From \Cref{prop:FirstWelfareTheorem}, the allocation is \fPO{}.} (These two properties are always maintained by the algorithm.) If the allocation at the end of Phase 1 is $\eps$-\EQ{1} with respect to the rounded instance $\I'$, then the algorithm terminates with this allocation as the output (Line~\ref{algline:TerminatePhase1}). Otherwise, it proceeds to Phase 2.

The allocation at the start of Phase 2 is not $\eps$-\EQ{1}, so there must exist an $\eps$-violator. Starting from the level $\ell = 1$ (Line~\ref{algline:InitializeLevel_Phase2}), the algorithm now performs a level-by-level search for an $\eps$-violator in the reachability set of the reference agent (Line~\ref{algline:IfCondition_Phase2}). As soon as an $\eps$-violator, say $h$, is found (along some alternating path $P$), the algorithm performs a pairwise swap between $h$ and the agent that precedes it along $P$ (Line~\ref{algline:Swap_Phase2}). Since the swapped chore is in the \MBB{} sets of both agents, the allocation continues to be \MBB{}-consistent after the swap. If, at any stage, the reference agent ceases to be the highest utility agent, Phase 2 restarts with the new reference agent (Line~\ref{algline:IdentityChange_Phase2}).

The above process continues until either the current allocation becomes $\eps$-\EQ{1} for the rounded instance $\I'$ (in which case the algorithm terminates and returns the current allocation as the output in Line~\ref{algline:TerminatePhase2}), or if no $\eps$-violator is reachable from the reference agent (Line~\ref{algline:WhileLoop_Phase2}). In the latter case, the algorithm proceeds to Phase 3.

Phase 3 involves uniformly lowering the prices of all the \emph{reachable chores}, i.e., the set of all chores that are collectively owned by all agents that are reachable from the reference agent (Line~\ref{algline:PriceDrop_Phase3}). The prices are lowered until a previously non-reachable agent becomes reachable due to the appearance of a new \MBB{} edge (Line~\ref{algline:Start_Of_Phase_3}). The algorithm now switches back to Phase 2 to start a fresh search for an $\eps$-violator in the updated reachability set (Line~\ref{algline:GoBackToPhase2_Phase3}).

\paragraph{Analysis of the algorithm}
The running time and correctness of our algorithm are established by \Cref{lem:RunningTime_ALG_EQ1+PO} and \Cref{lem:Correctness_ALG_EQ1+PO_original_instance} respectively, as stated below.

\begin{restatable}[\textbf{Running time}]{lemma}{RunningTimeALGEQonePO}
%\begin{lemma}
 \label{lem:RunningTime_ALG_EQ1+PO}
 Given as input any $\eps$-rounded instance with strictly negative valuations, \AlgEQonePO{} terminates in $\O(\poly(m,n,\ln |v_{\min}|,\nicefrac{1}{\eps}))$ time steps, where $v_{\min} = \min_{i,j} v_{i,j}$.
%\end{lemma}
\end{restatable}

The proof of \Cref{lem:RunningTime_ALG_EQ1+PO} appears in \Cref{subsec:Proof_RunningTime_ALG_EQ1+PO}.

\begin{restatable}[\textbf{Correctness}]{lemma}{CorrectnessALGEQonePOOriginal}
%\begin{lemma}
 \label{lem:Correctness_ALG_EQ1+PO_original_instance}
 Let $\I$ be any fair division instance with strictly negative valuations and $\I'$ be its $\eps$-rounded version for any given $\eps > 0$. Then, the allocation $A$ returned by \AlgEQonePO{} for the input $\I'$ is $3\eps$-\EQ{1} and $\eps$-\PO{} for $\I$. In addition, if $\eps \leq \frac{1}{6m |v_{\min}|^3}$, then $A$ is \EQ{1} and \PO{} for $\I$.
%\end{lemma}
\end{restatable}

The proof of \Cref{lem:Correctness_ALG_EQ1+PO_original_instance} appears in \Cref{subsec:Proof_Correctness_ALG_EQ1+PO_original_instance}.

Notice that the running time guarantee in \Cref{lem:RunningTime_ALG_EQ1+PO} is stated in terms of \emph{time steps}. A time step refers to a single iteration of Phase 1, Phase 2, or Phase 3. Since each individual iteration requires polynomial time, it suffices to analyze the running time of the algorithm in terms of the \emph{number} of iterations of the three phases.\footnote{Indeed, an iteration of Phase 1 involves assigning each chore to the agent with the highest valuation and setting its price. An iteration of Phase 2 involves the construction of the reachability set (say via breadth-first or depth-first search), followed by performing a level-wise search for an $\eps$-path-violator, followed by performing a swap operation. An iteration of Phase 3 involves scanning the set of reachable chores and setting an appropriate value of the price-drop factor $\Delta$. All of these operations can be carried out in $\O(\poly(m,n))$ time.} We will use the terms \emph{step}, \emph{time step}, and \emph{iteration} interchangeably.

\renewcommand{\floatpagefraction}{.95}
\begin{algorithm}%[H]
 \DontPrintSemicolon
 \linespread{1.2}\selectfont
 \KwIn{An $\eps$-rounded instance $\I' = \langle [n],[m],\W \rangle$.}
 \KwOut{An integral allocation $A$.}
% \Parameter{$0 < \eps < 1$.}
 \BlankLine
 \Comment{Phase 1: Initialization}
 \BlankLine
 \tikzmk{A}
 $A \leftarrow $ a utilitarian welfare-maximizing allocation\;
 \nonl (assign chore $j \in [m]$ to agent $i$ if $i \in \arg\max_{k \in [n]} w_{k,j}$)\label{algline:Assignment_Phase1}\;
 $\p \leftarrow $ For every chore $j \in [m]$, set $p_j = |w_{i,j}|$ if $j \in A_i$\label{algline:Prices_Phase1}\;
 \lIf{$A$ is $\eps$-\EQ{1} for $\I'$\label{algline:TerminatePhase1}}{\KwRet{$A$}}
 \nonl \tikzmk{B}
 \boxit{mygray}
 %\BlankLine
 \Comment{Phase 2: Remove \EQ{1} violations among the reachable agents}
 \BlankLine
 \oldnl \tikzmk{A}
 $i \leftarrow $ reference agent in $A$ \Comment*[r]{tiebreak lexicographically}\label{algline:Refresh_ReferenceAgent}
 $\R_i \leftarrow $ Reachability set of $i$ under $(A,\p)$\;
 $\ell = 1$ \Comment*[r]{initialize the level}\label{algline:InitializeLevel_Phase2}
 \While{$\R_i^{\ell}$ is non-empty and $A$ is not $\eps$-\EQ{1}\label{algline:WhileLoop_Phase2}}{
 		\uIf{$h \in \R_i^{\ell}$ is an $\eps$-path-violator along the alternating path $P = (i,j_1,h_1,\dots,j_{\ell - 1},h_{\ell - 1},j,h)$\label{algline:IfCondition_Phase2}}{$A_h \leftarrow A_h \setminus \{j\}$ and $A_{h_{\ell-1}} \leftarrow A_{h_{\ell-1}} \cup \{j\}$ \Comment*[r]{swap $j$}\label{algline:Swap_Phase2}
 			Repeat Phase 2 starting from Line~\ref{algline:Refresh_ReferenceAgent}\label{algline:IdentityChange_Phase2}
 		}\Else{$\ell \leftarrow \ell + 1$ \Comment*[r]{Move to the next level}}
 }\label{algline:MoveToNextLevel_Phase2}
  \lIf{$A$ is $\eps$-\EQ{1} for $\I'$\label{algline:TerminatePhase2}}{\KwRet{$A$}}
 %\BlankLine
 \nonl \tikzmk{B}
 \boxit{mygray}
 %\BlankLine
 \Comment{Phase 3: Price-drop}
  %\BlankLine
 \oldnl \tikzmk{A} 
 $\Delta \leftarrow \min\limits_{h \in \R_i, \, j \in [m] \setminus A_{\R_i}} \frac{w_{h,j}/p_j}{\beta_h}$, where $\beta_h$ is the \MBB{} ratio of $h$ (in $\I'$) and $A_{\R_i} \coloneqq \cup_{k \in \R_i} A_k$ is the set of reachable chores
 \label{algline:Start_Of_Phase_3}%\;
 \BlankLine
 \Comment{$\Delta $ is the smallest price-drop factor that makes a new agent reachable}
 \BlankLine
 \ForEach{chore $j \in A_{\R_i}$}{
 	$p_j \leftarrow p_j / \Delta$\Comment*[r]{uniformly lower the prices of reachable chores}
 \label{algline:PriceDrop_Phase3}}
 Repeat Phase 2 starting from Line~\ref{algline:Refresh_ReferenceAgent}\label{algline:GoBackToPhase2_Phase3}\;
 %\BlankLine
 \nonl \tikzmk{B}
 \boxit{mygray}
 \caption{\AlgEQonePO{}}
 \label{alg:EQ1+PO}
\end{algorithm}

We are now ready to prove \Cref{thm:EQ1+PO_pseudopolynomial}.

\EQonePOPseudopolynomial*

\begin{proof}
Fix $\eps = \frac{1}{6 m |v_{\min}|^3}$. Given a chores instance $\I$, its $\eps$-rounded version $\I'$ can be constructed in $\O( \poly(m,n,\ln |v_{\min}|) )$ time. We run the algorithm \AlgEQonePO{} on the input $\I'$. From \Cref{lem:RunningTime_ALG_EQ1+PO}, we know that the algorithm terminates in $\O(\poly(m,n,\ln |v_{\min}|,\nicefrac{1}{\eps}))$ time. \Cref{lem:Correctness_ALG_EQ1+PO_original_instance} implies that $A$ is \EQ{1} and \PO{} for $\I$.
\end{proof}

\subsection{Proof of Lemma~\ref{lem:RunningTime_ALG_EQ1+PO}}
\label{subsec:Proof_RunningTime_ALG_EQ1+PO}

Recall the statement of \Cref{lem:RunningTime_ALG_EQ1+PO}.
\RunningTimeALGEQonePO*
\begin{proof}
The proof of \Cref{lem:RunningTime_ALG_EQ1+PO} follows immediately from \Cref{lem:RunningTime_Phase2,lem:RunningTime_Phase3}, which are stated below.
\end{proof}

\begin{restatable}{lemma}{RunningTimePhaseTwo}
%\begin{lemma}
 \label{lem:RunningTime_Phase2}
 There can be at most $\O( \poly(m,n,\nicefrac{1}{\eps}) \ln m |v_{\min}| )$ consecutive iterations of Phase 2 before a Phase 3 step occurs.
%\end{lemma}
\end{restatable}

\begin{restatable}{lemma}{RunningTimePhaseThree}
%\begin{lemma}
 \label{lem:RunningTime_Phase3}
 There can be at most $\O( \poly(n,\nicefrac{1}{\eps})  \ln |v_{\min}| )$ Phase 3 steps during any execution of \AlgEQonePO{}.
%\end{lemma}
\end{restatable}

The proofs of \Cref{lem:RunningTime_Phase2,lem:RunningTime_Phase3} are provided in \Cref{subsec:Proof_RunningTime_Phase2,subsec:Proof_RunningTime_Phase3}, respectively.

\subsection{Proof of Lemma~\ref{lem:RunningTime_Phase2}}
\label{subsec:Proof_RunningTime_Phase2}

The proof of \Cref{lem:RunningTime_Phase2} relies on several intermediate results (\Cref{lem:Bound_On_Consecutive_Swaps,lem:Reference_Utility_Nonincreasing,lem:Lower_Bound_On_Increase_In_Utility,lem:Bound_On_Identity_Changes}) that are stated below.

\begin{restatable}{lemma}{BoundOnConsecutiveSwaps}
 \label{lem:Bound_On_Consecutive_Swaps}
 There can be at most $\O( \poly(m,n) )$ consecutive swap operations in Phase 2 before either the identity of the reference agent changes or a Phase 3 step occurs.
\end{restatable}

The proof of \Cref{lem:Bound_On_Consecutive_Swaps} is identical to \citep[Lemma 13]{BKV18Finding} and is therefore omitted.

Throughout, we will use the phrase \emph{at time step $t$} to refer to the state of the algorithm at the beginning of the time step $t$. In addition, we will use $i_t$ and $A^t \coloneqq (A^t_1,\dots,A^t_n)$ to denote the reference agent and the allocation maintained by the algorithm at the beginning of time step $t$, respectively. Thus, for instance, the utility of the reference agent at time step $t$ is $w_{i_t}(A^t_{i_t})$.

\begin{restatable}{lemma}{ReferenceUtilityNonincreasing}
%\begin{lemma}
 \label{lem:Reference_Utility_Nonincreasing}
 The utility of the reference agent cannot increase with time. That is, for any time step $t$, 
\begin{align*}
w_{i_t}(A^t_{i_t}) \geq w_{i_{t+1}}(A^{t+1}_{i_{t+1}}).
\end{align*}
%\end{lemma}
\end{restatable}
\begin{proof}
The only way in which the utility of a reference agent can change is via a swap operation in Phase 2. By construction, a reference agent can never lose a chore during a swap operation (though it can possibly receive a chore). Therefore, the utility of a reference agent cannot increase.
\end{proof}

\begin{restatable}{lemma}{LowerBoundOnIncreaseInUtility}
 \label{lem:Lower_Bound_On_Increase_In_Utility}
 Let $i$ be a fixed agent. Consider any set of consecutive Phase 2 steps during the execution of \AlgEQonePO{}. Suppose that $i$ turns from a reference to a non-reference agent during time step $t$. Let $t' > t$ be the first time step after $t$ at which $i$ once again becomes a reference agent. Then, either $A_i^t$ is a strict subset of $A_i^{t'}$ or $w_i(A_i^{t'}) < (1+\eps) w_i(A_i^t)$.
\end{restatable}
\begin{proof}
In order for a reference agent to turn into a non-reference agent, it must receive a chore during a swap operation. That is, agent $i$ must receive a chore at time $t$ and hence $A_i^t$ is a strict subset of $A_i^{t+1}$. If agent $i$ does not lose any chore between $t+1$ and $t'$, then the claim follows. Therefore, for the rest of the proof, we will assume that agent $i$ loses at least one chore between $t+1$ and $t'$.

Among all the time steps between $t+1$ and $t'$ at which agent $i$ loses a chore, let $\tau$ be the last one. Let $i_\tau$ be the reference agent at time step $\tau$. Since the utility of the reference agent is non-increasing with time (\Cref{lem:Reference_Utility_Nonincreasing}), we have that
\begin{align}
w_{i_\tau}(A_{i_\tau}^\tau) \leq w_i(A_i^t).
\label{eqn:LowerBoundOnIncreaseInUtility_1}
\end{align}
Let $c$ denote the chore that agent $i$ loses at time step $\tau$. An agent that loses a chore must be an $\eps$-path violator (with respect to an alternating path involving that chore). Therefore,
\begin{align}
w_i(A_i^\tau \setminus \{c\}) < (1+\eps) w_{i_\tau}(A_{i_\tau}^\tau).
\label{eqn:LowerBoundOnIncreaseInUtility_2}
\end{align}
Since $i$ does not lose any chore between $\tau$ and $t'$, we have
\begin{align}
w_i(A_i^{t'}) \leq w_i(A_i^{\tau+1}) = w_i(A_i^\tau \setminus \{c\}).
\label{eqn:LowerBoundOnIncreaseInUtility_3}
\end{align}
Combining \Cref{eqn:LowerBoundOnIncreaseInUtility_1,eqn:LowerBoundOnIncreaseInUtility_2,eqn:LowerBoundOnIncreaseInUtility_3} gives 
\begin{align*}
w_i(A_i^{t'}) < (1+\eps) w_i(A_i^t),
\end{align*}
 as desired.
\end{proof}

\begin{restatable}{lemma}{BoundOnIdentityChanges}
 \label{lem:Bound_On_Identity_Changes}
 There can be at most $\O( \poly(m,n,\nicefrac{1}{\eps}) \ln m |v_{\min}| )$ changes in the identity of the reference agent before a Phase 3 step occurs.
\end{restatable}
\begin{proof}
From \Cref{lem:Lower_Bound_On_Increase_In_Utility}, we know that each time the algorithm cycles back to a some agent $i$ as the reference agent, either the allocation of agent $i$ grows strictly by at least one chore, or its utility decreases by at least a multiplicative factor of $(1+\eps)$. By pigeonhole principle, after every $n$ consecutive changes in the identity of the reference agent, the algorithm must cycle back to some agent as the reference. Along with the fact that the utility of the reference agent is non-increasing with time (\Cref{lem:Reference_Utility_Nonincreasing}), we get that after every $mn$ consecutive identity changes, the utility of the reference agent must decrease multiplicatively by a factor of $(1+\eps)$. Since there are $m$ chores overall, the utility of any agent can never be less than $m w_{\min}$. Hence, there can be at most $mn \log_{1+\eps} m |w_{\min}|$ changes in the identity of the reference agent during the execution of the algorithm. The stated bound now follows from $\eps$-roundedness and the fact that $\frac{1}{\ln (1+\eps)} \leq \frac{2}{\eps}$ for every $\eps \in (0,1)$.
\end{proof}

We are now ready to prove \Cref{lem:RunningTime_Phase2}.
\RunningTimePhaseTwo*

\begin{proof}
From \Cref{lem:Bound_On_Identity_Changes}, we know that there can be at most $\O( \poly(m,n,\nicefrac{1}{\eps}) \ln m |v_{\min}| )$ changes in the identity of the reference agent (in Phase 2) before a Phase 3 step occurs. Furthermore, \Cref{lem:Bound_On_Consecutive_Swaps} implies that there can be at most $\O( \poly(m,n) )$ swap operations between two consecutive identity changes or an identity change and a Phase 3 step. Combining these implications gives the desired bound. 
\end{proof}

\subsection{Proof of Lemma~\ref{lem:RunningTime_Phase3}}
\label{subsec:Proof_RunningTime_Phase3}

The proof of \Cref{lem:RunningTime_Phase3} relies on several intermediate results (\Cref{lem:Set_Of_Violators_Cannot_Grow,lem:Allocation_Of_Violators_Cannot_Grow,lem:Violators_Are_Not_Reachable,lem:MBB_Lower_Bound_Violator_Good,cor:MBB_Lower_Bound_Implication}) that are stated and proved below. It will be useful to define the set $E_t$ of all $\eps$-violators at time step $t$. That is,
$$E_t \coloneqq \{k \in [n] : w_k(A^t_k \setminus \{j\}) < (1+\eps) w_{i_t}(A^t_{i_t}) \, \forall j \in A^t_k \},$$
where $i_t$ is the reference agent at time step $t$.

Some of our proofs will require the following assumption:
\begin{assumption}
	At the end of Phase 1 of \AlgEQonePO{}, every agent is assigned at least one chore.
\label{assumption:Each_agent_gets_a_chore}
\end{assumption}
This assumption can be ensured via efficient preprocessing techniques similar to those used by \citet{BKV18Finding}. We refer the reader to Section B.1 of their paper for details.

\begin{restatable}{lemma}{SetOfViolatorsCannotGrow}
 \label{lem:Set_Of_Violators_Cannot_Grow}
 Let $t$ and $t'$ be two Phase 3 time steps such that $t < t'$. Then, $E_{t'} \subseteq E_t$.
\end{restatable}
\begin{proof}
It suffices to consider consecutive Phase 3 steps $t$ and $t'$ such that all intermediate time steps $t+1, t+2, \dots, t'-1$ occur in Phase 2. Suppose, for contradiction, that there exists some agent $k \in E_{t'} \setminus E_t$. Observe that a non-$\eps$-violator cannot turn into an $\eps$-violator in Phase 3 as the allocation of the chores remains fixed during price-drop. Therefore, the only way in which $k$ can turn into an $\eps$-violator is via a swap operation in Phase 2. In the rest of the proof, we will argue that if there is a swap operation at time step $\tau$ (where $t < \tau < t'$) that turns $k$ into an $\eps$-violator, then there is a subsequent swap operation at time step $\tau + 1$ that turns it back into a non-$\eps$-violator. This will provide the desired contradiction.

Suppose that agent $k$ is at level $\ell$ in the reachability set when it receives a chore $c$ that turns it into an $\eps$-violator. Recall that a swap operation involves transferring a chore from an agent at a higher level $\ell + 1$ to one at a lower level $\ell$. Furthermore, a swap involving an agent at level $\ell + 1$ happens only when no agent in the levels $1,2,\dots,\ell$ is an $\eps$-path violator. Therefore, agent $k$ cannot be an $\eps$-path violator just before the time step $\tau$. In other words, there must exist a chore $c'$ on an alternating path from the reference agent $i_\tau$ to agent $k$ such that
\begin{align}
(1+\eps) w_{i_\tau}(A^\tau_{i_\tau}) \leq w_k(A^\tau_k \setminus \{c'\}).
\label{eqn:Set_Of_Violators_Cannot_Grow_1}
\end{align}
Since agent $k$ becomes an $\eps$-violator (and hence an $\eps$-path violator) after receiving the chore $c$, we have
\begin{align*}
w_k(A^\tau_k \cup \{c\} \setminus \{c'\}) & < (1+\eps) w_{i_{\tau+1}}(A^{\tau+1}_{i_{\tau+1}}) \\
& = (1+\eps) w_{i_\tau}(A^\tau_{i_\tau}),
\end{align*}
where the equality follows from the observation that neither the identity nor the allocation of the reference agent changes during the above swap. Note that the swap involving $c$ does not affect the alternating path to agent $k$ that includes the chore $c'$. This means that agent $k$ now becomes the \emph{only} $\eps$-path-violator at level $\ell$ or below. Therefore, in a subsequent swap operation at time step $\tau + 1$, the algorithm will take $c'$ away from agent $k$, resulting in a new bundle $A_k^{\tau+1} = A^\tau_k \cup \{c\} \setminus \{c'\}$. From \Cref{eqn:Set_Of_Violators_Cannot_Grow_1}, we get that agent $k$ is a non-$\eps$-violator up to the removal of the chore $c$, as desired.
\end{proof}

\begin{restatable}{lemma}{AllocationOfViolatorsCannotGrow}
 \label{lem:Allocation_Of_Violators_Cannot_Grow}
 Let $t$ and $t'$ be two Phase 3 time steps such that $t < t'$. Then, for any $k \in E_{t'}$, $A^{t'}_k \subseteq A^t_k$.
\end{restatable}
\begin{proof}(Sketch.)
Suppose, for contradiction, that there exists a chore $c \in A^{t'}_k \setminus A^t_k$. The only way in which agent $k$ could have acquired the chore $c$ is via a swap operation at time step $\tau$ for some $t < \tau < t'$. Thus, agent $k$ cannot be an $\eps$-path-violator just before the time step $\tau$, and therefore also cannot be an $\eps$-violator. By an argument similar to that in the proof of \Cref{lem:Set_Of_Violators_Cannot_Grow}, it follows that agent $k$ cannot be an $\eps$-violator at time step $t'$, giving us the desired contradiction.
\end{proof}

\begin{restatable}{lemma}{ViolatorsAreNotReachable}
 \label{lem:Violators_Are_Not_Reachable}
 For any Phase 3 time step $t$, $E_t \cap \R_{i_t} = \emptyset$.
\end{restatable}
\begin{proof} 
Suppose, for contradiction, that there exists some $k \in E_t \cap \R_{i_t}$ at time step $t$, i.e., $k$ is an $\eps$-violator that is reachable (via some alternating path). Then, agent $k$ must also be an $\eps$-path violator, implying that the algorithm continues to be in Phase 2 at time step $t$ and therefore cannot enter Phase 3.
\end{proof}

\begin{restatable}{lemma}{MBBLowerBound}
 \label{lem:MBB_Lower_Bound_Violator_Good}
 Let $t$ be a Phase 3 time step. Then, there exists an $\eps$-violator $k \in E_t$ and a chore $j \in A_k^t$ such that for every agent $i \in [n]$, $\beta_i^t \geq \nicefrac{w_{i,j}}{|w_{k,j}|}$, where $\beta_i^t$ is the \MBB{} ratio of agent $i$ at time step $t$.
\end{restatable}
\begin{proof}
Note that the algorithm enters Phase 3 at time step $t$ only if the current allocation $A^t$ is not $\eps$-\EQ{1}. Thus, there must exist an $\eps$-violator agent $k \in E_t$. Fix any chore $j \in A^t_k$ (this is well-defined since $A^t_k \neq \emptyset$). From \Cref{lem:Set_Of_Violators_Cannot_Grow,lem:Allocation_Of_Violators_Cannot_Grow}, we know that $k \in E_{\tau}$ and $j \in A^{\tau}_k$ for all Phase 3 time steps $\tau < t$. Additionally, for every Phase 3 time step $\tau$ preceding the time step $t$, we know from \Cref{lem:Violators_Are_Not_Reachable} that $k \notin R_{i_{\tau}}$. In other words, the agent $k$ never experiences a price-drop between the start of the algorithm and the time step $t$. As a result, the \MBB{} ratio of agent $k$ at time step $t$ is the same as that at the time of the first price-drop, i.e., $\beta^t_k = \beta^{t_1}_k$, where $t_1$ denotes the earliest Phase 3 time step. Furthermore, since the \MBB{} ratios of all agents remain unchanged during Phase 2, we must have that $\beta^{t_1}_k = -1$ (this follows from the way we set the initial prices in Phase 1), and thus also $\beta^t_k = -1$. By a similar argument, the chore $j$ does not experience a price-drop between the start of the algorithm and the time step $t$. Therefore, $p^t_j = p^{t_1}_j$. Since the allocation maintained by the algorithm is always \MBB{}-consistent, we get that $p^{t_1}_j = |w_{k,j}|$. The claim now follows by noticing that each agent's \MBB{} ratio is at least its bang-per-buck ratio for the chore $j$.
\end{proof}

\begin{restatable}{corollary}{MBBLowerBoundImplication}
 \label{cor:MBB_Lower_Bound_Implication}
 Let $t$ be a Phase 3 time step. Then, for every agent $i \in [n]$, we have $\beta_i^t \geq w_{\min}$, where $\beta_i^t$ is the \MBB{} ratio of agent $i$ at time step $t$ and $w_{\min} = \min_{i,j} w_{i,j}$.
\end{restatable}
\begin{proof}
Let $k \in E_t$ be an $\eps$-violator at time $t$, and let $j \in A_k^t$ be a chore owned by $k$. From \Cref{lem:MBB_Lower_Bound_Violator_Good}, we know that for every agent $i \in [n]$, $\beta_i^t \geq \nicefrac{w_{i,j}}{|w_{k,j}|}$, where $\beta_i^t$ is the \MBB{} ratio of agent $i$ at time step $t$. Since $w_{i,j} \geq w_{\min}$, we get that $\beta_i^t \geq \nicefrac{w_{\min}}{|w_{k,j}|} = -\nicefrac{|w_{\min}|}{|w_{k,j}|}$. 

By assumption, all valuations in the original instance $\I$ are strictly negative and integral. This means that in the $\eps$-rounded version $\I'$, for every $i \in [n]$ and $j \in [m]$, we have $-1 \geq v_{i,j} \geq w_{i,j}$. Therefore, $|w_{k,j}| \geq 1$, or, equivalently, $\frac{-1}{|w_{k,j}|} \geq -1$. Using this bound in the above inequality, we get that $\beta_i^t \geq -|w_{\min}| = w_{\min}$, as desired.
\end{proof}

We are now ready to prove \Cref{lem:RunningTime_Phase3}.
\RunningTimePhaseThree*

\begin{proof}
The proof uses a potential function argument. For any Phase 3 time step $t$, we define a potential 
\begin{align*}
\Phi^t \coloneqq \sum_{i \in [n]} \log_{1+\eps} |\beta_i^t|,
\end{align*}
where $\beta_i^t \coloneqq \max_{j \in [m]} \nicefrac{w_{i,j}}{p^t_j}$ is the \MBB{} ratio of agent $i$ and $p^t_j$ is the price of chore $j$ at time step $t$. 

In Phase 1, the price of every chore is set to be the absolute value of the highest valuation for that chore. Along with Assumption~\ref{assumption:Each_agent_gets_a_chore}, this implies that at the end of Phase 1, the \MBB{} ratio of every agent equals $-1$. Since Phase 2 does not affect the prices, the \MBB{} ratio of every agent at the time of the earliest price-drop also equals $-1$. Thus, the initial value of the potential $\Phi^1$ is $0$.

We will now argue that each time the algorithm performs a price-drop, the potential must increase by at least $1$ (i.e., for any two Phase 3 steps $t$ and $t'$ such that $t < t'$, $\Phi^{t'} - \Phi^t \geq 1$). Recall that the valuations are strictly negative. Also, the prices are always strictly positive, and are non-increasing with time. Therefore, all bang-per-buck ratios (and hence all \MBB{} ratios) are always strictly negative and are non-increasing with time. Consequently, for any agent $i$, $|\beta_i^t|$ is non-decreasing with time. Thus, $\Phi^t \geq 0$ for all time steps $t \in \{1,2,\dots\}$. In addition, each time the algorithm performs a price-drop, the \MBB{} ratio of some agent \emph{strictly} decreases (because a new chore gets added to the \MBB{} set of some agent).

We will argue that the (multiplicative) drop in \MBB{} ratio is always by a positive integral power of $(1+\eps)$. Indeed, by assumption, all valuations are (negative of) integral powers of $(1+\eps)$. We observed earlier that all \MBB{} ratios at the end of Phase 1 are equal to $-1$, which means that all initial prices must be integral powers of $(1+\eps)$. Furthermore, the price-drop factor $\Delta$ is a ratio of bang-per-buck ratios, and is therefore also an integral power of $(1+\eps)$. So, whenever the \MBB{} ratio of some agent strictly decreases, it must do so by an integral power of $(1+\eps)$. After each price-drop in Phase 3, the potential must therefore increase by at least $1$.

All that remains to be shown is an upper bound on the potential $\Phi^t$. From \Cref{cor:MBB_Lower_Bound_Implication}, we know that for every Phase 3 time step $t$, we have $\beta^t_i \geq w_{\min}$, and consequently, $\Phi^t \leq n  \log_{1+\eps} {|w_{\min}|}$. Since the potential increases by at least $1$ between any consecutive price-drops, the overall number of Phase 3 time steps can be at most $n \log_{1+\eps} |w_{\min}|$. For $\eps$-rounded valuations, we have $|w_{\min}| \leq (1+\eps) |v_{\min}|$, and therefore $n \log_{1+\eps} |w_{\min}| = n + n \log_{1+\eps} |v_{\min}|$. The stated bound now follows by observing that $\frac{1}{\ln (1+\eps)} \leq \frac{2}{\eps}$ for every $\eps \in (0,1)$.
\end{proof}

\subsection{Proof of Lemma~\ref{lem:Correctness_ALG_EQ1+PO_original_instance}}
\label{subsec:Proof_Correctness_ALG_EQ1+PO_original_instance}

The proof of \Cref{lem:Correctness_ALG_EQ1+PO_original_instance} relies on several intermediate results (\Cref{lem:Correctness_ALG_EQ1+fPO_rounded_instance,lem:fPO_Rounded_eps_PO_Original,lem:Bound_On_Eps_For_Exact_PO,lem:epsEQ1_Rounded_3epsEQ1_Original,lem:Bound_On_eps_for_Exact_EQone}) as stated below.

\begin{restatable}{lemma}{CorrectnessALGEQonefPORounded}
%\begin{lemma}
 \label{lem:Correctness_ALG_EQ1+fPO_rounded_instance}
 Given as input any $\eps$-rounded instance $\I'$ with strictly negative valuations, the allocation $A$ returned by \AlgEQonePO{} is $\eps$-\EQ{1} and \fPO{} for $\I'$.
%\end{lemma}
\end{restatable}
\begin{proof}
From \Cref{lem:RunningTime_ALG_EQ1+PO}, we know that \AlgEQonePO{} is guaranteed to terminate. Furthermore, the algorithm can only terminate in Lines~\ref{algline:TerminatePhase1} or \ref{algline:TerminatePhase2}. In both cases, the allocation $A$ returned by the algorithm is guaranteed to be $\eps$-\EQ{1} with respect to the input instance $\I'$.

To see why $A$ is \fPO{}, note that \AlgEQonePO{} always maintains an \MBB{}-consistent allocation (with respect to the current prices). Define a Fisher market where each agent is assigned a budget equal to its spending under $A$. Then, the outcome $(A,\p)$ satisfies the equilibrium conditions for this market. Therefore, from \Cref{prop:FirstWelfareTheorem}, $A$ is \fPO{}.
\end{proof}

\begin{restatable}{lemma}{fPOForRoundedIsEpsPOForOriginal}
%\begin{lemma}
 \label{lem:fPO_Rounded_eps_PO_Original}
 Let $\I$ be any fair division instance and $\I'$ be its $\eps$-rounded version for any given $\eps > 0$. Then, an allocation $A$ that is \fPO{} for $\I'$ is $\eps$-\PO{} for $\I$.
%\end{lemma}
\end{restatable}
\begin{proof}
Suppose, for contradiction, that $A$ is $\eps$-Pareto dominated in $\I$ by an allocation $B$, i.e., $v_k(B_k) \geq \frac{1}{(1+\eps)} v_k(A_k)$ for every agent $k \in [n]$ and $v_i(B_i) > \frac{1}{(1+\eps)} v_i(A_i)$ for some agent $i \in [n]$. Since $\I'$ is an $\eps$-rounded version of $\I$, we have that $v_{i,j} \geq w_{i,j} \geq (1+\eps)v_{i,j}$ for every agent $i$ and every chore $j$. Using this bound and the additivity of valuations, we get that $w_k(B_k) \geq w_k(A_k)$ for every agent $k \in [n]$ and $w_i(B_i) > w_i(A_i)$ for some agent $i \in [n]$. Thus, $B$ Pareto dominates $A$ in the instance $\I'$, which is a contradiction since $A$ is \fPO{} (hence \PO{}) for $\I'$.
\end{proof}

\begin{restatable}{lemma}{BoundOnMBBAfterTermination}
%\begin{lemma}
 \label{lem:Bound_On_MBB_After_Termination}
 Let $\p$ denote the price-vector right after the termination of \AlgEQonePO{}. Let $\beta_k$ denote the \MBB{} ratio of agent $k \in [n]$ in $\I'$ with respect to $\p$. Then, $\beta_k \geq -|w_{\min}|^2$.
%\end{lemma}
\end{restatable}
\begin{proof}
Let $t_1,\dots,t_N$ denote the Phase 3 time steps during the execution of the algorithm. Let $\beta_k^{t_N}$ denote the \MBB{} ratio of agent $k$ \emph{before} the price-drop at $t_N$ takes place, and let $\Delta_N$ denote the (multiplicative) price-drop factor at time step $t_N$. In addition, let $i$ denote the reference agent at $t_N$. For every agent $k$ that is reachable at $t_N$ (i.e., $k \in \R_i^{t_N}$), we have that $\beta_k = \beta_k^{t_N} \cdot \Delta_N$. Similarly, for every $k \notin \R_i^{t_N}$, $\beta_k = \beta_k^{t_N}$.

We know from \Cref{cor:MBB_Lower_Bound_Implication} that $\beta_k^{t_N} \geq w_{\min}$. Since $\beta_k^{t_N} < 0$, it suffices to prove that $\Delta_N \leq |w_{\min}|$.

By definition, $\Delta_N \leq \frac{w_{h,j}/p_j^{t_N}}{\beta_h^{t_N}}$ for every agent $h \in \R_i$ and every chore $j \in [m] \setminus A_{\R_i}$; here, $A_{\R_i} \coloneqq \cup_{h \in \R_i} A_h$ is the set of reachable chores at $t_N$. Recall from the proof of \Cref{lem:RunningTime_Phase3} that all \MBB{} ratios are initially equal to $-1$ and are non-increasing with time. Thus, $\beta_h^{t_N} \leq -1$. This implies that
\begin{align*}
\Delta_N \leq \frac{w_{h,j}/p_j^{t_N}}{\beta_h^{t_N}} = \frac{|w_{h,j}|}{p_j^{t_N} \cdot |\beta_h^{t_N}|} \leq \frac{|w_{\min}|}{p_j^{t_N}}.
\end{align*}
The above inequality holds for \emph{every} chore $j$ that is not reachable at time step $t_N$. In particular, we can choose $j$ to any chore owned by an $\eps$-violator $k$ at $t_N$. Note that our choice of $j$ is well-defined: Indeed, such an agent $k$ must exist because the allocation is not $\eps$-\EQ{1} when the algorithm enters Phase 3 at $t_N$. Furthermore, since agent $k$ is an $\eps$-violator, it must own at least one chore.

By an argument similar to that in the proof of \Cref{lem:MBB_Lower_Bound_Violator_Good}, we get that $p_j^{t_N} = p_j^{t_1} = |w_{k,j}| \geq 1$, where the inequality follows from the integrality of valuations. Substituting $p_j^{t_N} \geq 1$ gives $\Delta_N \leq |w_{\min}|$, as desired.
\end{proof}

\begin{restatable}{lemma}{BoundOnEpsForExactPO}
%\begin{lemma}
 \label{lem:Bound_On_Eps_For_Exact_PO}
 Let $\I$ be any fair division instance and $\I'$ be its $\eps$-rounded version for any $0 < \eps \leq \frac{1}{6m |v_{\min}|^3}$. Let $A$ be the allocation returned by \AlgEQonePO{} for the input instance $\I'$. Then, $A$ is \PO{} for $\I$.
%\end{lemma}
\end{restatable}
\begin{proof}
Suppose, for contradiction, that the allocation $A$ is Pareto dominated by an allocation $B$ in the instance $\I$. That is, $v_k(B_k) \geq v_k(A_k)$ for every agent $k \in [n]$ and $v_i(B_i) > v_i(A_i)$ for some agent $i \in [n]$. Since the valuations in $\I$ are integral, we have $v_i(B_i) \geq v_i(A_i) + 1$. 

Let $\p$ denote the price-vector right after the termination of \AlgEQonePO{}. Let $\alpha_k$ and $\beta_k$ denote the \MBB{} ratios (with respect to $\p$) of agent $k \in [n]$ in $\I$ and $\I'$ respectively. That is,
$\alpha_k = \max_{j \in [m]} v_{k,j}/p_j$ and $\beta_k = \max_{j \in [m]} w_{k,j}/p_j$. Since $\I'$ is an $\eps$-rounded version of $\I$, we have $v_{k,j} \geq w_{k,j} \geq (1+\eps)v_{k,j}$ for every agent $k \in [n]$ and every chore $j \in [m]$. Thus,
$\alpha_k \geq \beta_k \geq (1+\eps) \alpha_k$. Now consider the allocation $B$. By definition of \MBB{} ratio, we have that for every $k \in [n]$, $\beta_k \p(B_k) \geq w_k(B_k)$, or, equivalently, $\p(B_k) \leq \frac{w_k(B_k)}{\beta_k}$ (since $\beta_k < 0$).

The combined spending over all the chores is given by
\begin{align*}
 \p([m]) & = \sum_{k \in [n]} \p(B_k) & \text{(since all the chores are allocated under $B$)}\\
  & = \p(B_i) + \sum_{k \in [n] \setminus \{i\}} \p(B_k) &\\
  & \leq \frac{w_i(B_i)}{\beta_i} + \sum_{k \neq i} \frac{w_k(B_k)}{\beta_k} &\\
  & \leq (1+\eps) \left( \frac{v_i(B_i)}{\beta_i} + \sum_{k \neq i} \frac{v_k(B_k)}{\beta_k} \right) & \text{(using $w_k(A_k) \geq (1+\eps) v_k(A_k)$ and $\beta_k < 0$)}\\
  & \leq (1+\eps) \left( \frac{v_i(A_i)+1}{\beta_i} + \sum_{k \neq i} \frac{v_k(A_k)}{\beta_k} \right) & \text{(using Pareto dominance and $\beta_k < 0$)}\\
  & = (1+\eps) \left( \frac{1}{\beta_i} + \sum_{k} \frac{v_k(A_k)}{\beta_k} \right) &\\
  & \leq (1+\eps) \left( \frac{1}{\beta_i} + \sum_{k} \frac{w_k(A_k)}{\beta_k} \right) & \text{(using $v_k(A_k) \geq w_k(A_k)$ and $\beta_k < 0$)}\\
  & = (1+\eps) \left( \frac{1}{\beta_i} + \p([m]) \right) & \text{(since $A$ is \MBB-consistent in $\I'$)}.
\end{align*}
Simplifying the above relation gives
\begin{alignat*}{3}
&& (1+\eps) \cdot \frac{-1}{\beta_i} & \leq \eps (\p([m]) &\\
\Rightarrow && (1+\eps) \cdot \frac{1}{|w_{\min}|^2} & \leq \eps \cdot m \cdot |w_{\min}| & \text{(using \Cref{lem:Bound_On_MBB_After_Termination})}\\
\Rightarrow && \frac{1}{m \cdot |w_{\min}|^3} & \leq \frac{\eps}{1+\eps} &\\
\Rightarrow && \frac{1}{m \cdot |v_{\min}|^3 \cdot {(1+\eps)}^3} & \leq \frac{\eps}{1+\eps} & \text{(using $\eps$-roundedness)}\\
\Rightarrow && \frac{1}{m \cdot |v_{\min}|^3} & \leq \eps {(1+\eps)}^2 & \\
\Rightarrow && \frac{1}{m \cdot |v_{\min}|^3} & \leq 4\eps & \text{(since $(1+\eps)^2 < 4$ for $\eps \in (0,1)$)},
\end{alignat*}
which contradicts the assumed bound on $\eps$. Thus, $A$ is \PO{} for $\I$.
\end{proof}

\begin{restatable}{lemma}{EpsEQ1RoundedThreeEpsEQ1Original}
%\begin{lemma}
 \label{lem:epsEQ1_Rounded_3epsEQ1_Original}
 Let $\I$ be any fair division instance and $\I'$ be its $\eps$-rounded version for any given $\eps > 0$. Then, an allocation $A$ that is $\eps$-\EQ{1} for $\I'$ is $3\eps$-\EQ{1} for $\I$.
%\end{lemma}
\end{restatable}
\begin{proof}
Since $A$ is $\eps$-\EQ{1} for $\I'$, we have that for every pair of agents $i,k \in [n]$, either $\frac{1}{(1+\eps)} w_k(A_k) \geq w_i(A_i)$ or there exists a chore $j \in A_k$ such that $\frac{1}{(1+\eps)} w_k(A_k \setminus \{j\}) \geq w_i(A_i)$. By $\eps$-roundedness, we have that $v_{i,j} \geq w_{i,j} \geq (1+\eps)v_{i,j}$ for every agent $i$ and every chore $j$. Along with the additivity of valuations (in $\I'$), this implies that for every pair of agents $i,k \in [n]$, either $v_k(A_k) \geq (1+\eps)^2 v_i(A_i)$ or there exists a chore $j \in A_k$ such that $v_k(A_k \setminus \{j\}) \geq (1+\eps)^2 v_i(A_i)$.

Since $0 < \eps < 1$, we have that $\eps^2 < \eps$. Furthermore, since $v_i(A_i) < 0$, we have that $\eps^2 v_i(A_i) > \eps v_i(A_i)$. Substituting this in the above inequalities gives us that for every pair of agents $i,k \in [n]$, either $v_k(A_k) \geq (1+3\eps) v_i(A_i)$ or there exists a chore $j \in A_k$ such that $v_k(A_k \setminus \{j\}) \geq (1+3\eps) v_i(A_i)$, as desired.
\end{proof}

\begin{restatable}{lemma}{BoundOnEpsforExactEQone}
%\begin{lemma}
 \label{lem:Bound_On_eps_for_Exact_EQone}
 Given any fair division instance $\I$ and any $0 \leq \eps \leq \frac{1}{6 m |v_{\min}|}$, an allocation $A$ is $3\eps$-\EQ{1} for $\I$ if and only if it is \EQ{1} for $\I$.
%\end{lemma}
\end{restatable}
\begin{proof}
Since $A$ is $3\eps$-\EQ{1} for $\I$, we have that for every pair of agents $i,k \in [n]$, either $v_k(A_k) \geq (1+3\eps) v_i(A_i)$ or there exists a chore $j \in A_k$ such that $v_k(A_k \setminus \{j\}) \geq (1+3\eps) v_i(A_i)$. Using the bound $\eps \leq \frac{1}{6 m |v_{\min}|}$, we get that either $v_i(A_i) - v_k(A_k) \leq \frac{1}{2}$ or there exists a chore $j \in A_k$ such that $v_i(A_i) - v_k(A_k \setminus \{j\}) \leq \frac{1}{2}$. Since the valuations are integral, this implies that either $v_i(A_i) - v_k(A_k) \leq 0$ or there exists a chore $j \in A_k$ such that $v_i(A_i) - v_k(A_k \setminus \{j\}) \leq 0$, which is the desired \EQ{1} condition.
\end{proof}

We are now ready to prove \Cref{lem:Correctness_ALG_EQ1+PO_original_instance}.
\CorrectnessALGEQonePOOriginal*

\begin{proof} 
The allocation $A$ returned by \AlgEQonePO{} is guaranteed to be $\eps$-\EQ{1} and \fPO{} with respect to the input instance $\I'$ (\Cref{lem:Correctness_ALG_EQ1+fPO_rounded_instance}).  \Cref{lem:epsEQ1_Rounded_3epsEQ1_Original,lem:fPO_Rounded_eps_PO_Original} together imply that $A$ is $3\eps$-\EQ{1} and $\eps$-\PO{} for $\I$. Furthermore, if $\eps \leq \frac{1}{6m |v_{\min}|^3}$, then the bounds in \Cref{lem:Bound_On_Eps_For_Exact_PO,lem:Bound_On_eps_for_Exact_EQone} are satisfied, which implies that $A$ is \EQ{1} and \PO{} for $\I$.
\end{proof}

% \clearpage 

\subsection{\EQX{} and \DEQX{} without the \texorpdfstring{$v_{i,j}<0$}{Negative Value} Condition}
\label{subsec:Zero-Valued-Chores-Removal-Variants}

In this section, we will consider alternative versions of \EQX{} and \DEQX{} notions wherein the $v_{i,j}<0$ condition is removed. 

Formally, an allocation $A$ is \emph{equitable up to any possibly zero-valued chore} (\EQXzero{}) if for every pair of agents $i,k \in [n]$ such that $A_i \neq \emptyset$ and for every chore $j \in A_i$, we have $v_i(A_i \setminus \{j\}) \geq v_k(A_k)$. Similarly, an allocation $A$ is \emph{equitable up to any possibly zero-valued duplicated chore} (\DEQXzero{}) if for every pair of agents $i,k \in [n]$ such that $A_i \neq \emptyset$ and for every chore $j \in A_i$, we have $v_i(A_i) \geq v_k(A_k \cup \{j\})$. Notice that an \EQXzero{} (respectively, \DEQXzero{}) allocation is also \EQX{} (respectively, \DEQX{}). The converse might not be true in general, but it does hold when agents have strictly negative valuations (i.e., $v_{i,j} < 0$ for all $i \in [n], j \in [m]$).

Recall from \Cref{thm:EQX_PO_Strong_NP-hardness} that checking the existence of \EQX{}+\PO{} is strongly \NPH{} even for strictly negative valuations. It readily follows that the same is true for \EQXzero{}+\PO{} allocations as well.

\begin{restatable}[\textbf{Hardness of \EQXzero{}+\PO{}}]{corollary}{EQXzeroPOStronglyNPHard}
%  \label{cor:EQXzero_PO_strongly_NPHard}
Determining whether a given instance admits an allocation that is equitable up to any possibly zero-valued chore $(\EQXzero{})$ and Pareto optimal $(\PO{})$ is strongly \NPH{}.
\end{restatable}

We know from \Cref{prop:DEQX+PO} that a \DEQX{}+\PO{} allocation is guaranteed to exist. By contrast, a \DEQXzero{}+\PO{} allocation might not always exist, and checking the existence of such allocations is strongly \NPH{}.

\begin{restatable}[\textbf{Hardness of \DEQXzero{}+\PO{}}]{lemma}{DEQXzeroPOStronglyNPHard}
%\begin{lemma}
 \label{lem:DEQXzero_PO_strongly_NPHard}
 Determining whether a given instance admits an allocation that is equitable up to any possibly zero-valued duplicated chore $(\DEQXzero{})$ and Pareto optimal $(\PO{})$ is strongly \NPH{}.
%\end{lemma}
\end{restatable}
\begin{proof}
We will show a reduction from \ThreePartition{}. An instance of \ThreePartition{} consists of a set $X = \{b_1,\dots,b_{3r}\}$ of $3r$ positive integers where $r \in \N$, and the goal is to find a partition of $X$ into $r$ subsets $X^1,\dots,X^r$ such that the sum of numbers in each subset is equal to $B \coloneqq \frac{1}{r} \sum_{b_i \in X} b_i$.\footnote{We do not require the sets $X^1,\dots,X^r$ to be of cardinality three each; \ThreePartition{} remains strongly \NPhard{} even without this constraint.} We will assume, without loss of generality, that for every $i \in [3r]$, $b_i$ is even and $b_i \geq 2$.

We will construct a chore division instance with $r$ agents $a_1,\dots,a_r$ and $4r$ chores $c_1,\dots,c_{4r}$. For every $i \in [r]$ and $j \in [3r]$, agent $a_i$ values chore $c_j$ at $-b_j$. In addition, agent $a_i$ values the chore $c_{3r+i}$ at $0$, and every other chore in $\{c_{3r+1},\dots,c_{4r}\} \setminus \{c_{3r+i}\}$ at $-1$.

($\Rightarrow$) Suppose $X^1,\dots,X^r$ is a solution of \ThreePartition{}. Then, we can construct a perfectly equitable (and therefore also \DEQXzero{}) and Pareto optimal allocation as follows: For every $i \in [r]$ and $j \in [3r]$, $A_i = \{c_j \, : \, b_j \in X^i\} \cup \{c_{3r+i}\}$. 

($\Leftarrow$) Now suppose there exists an allocation $A$ that is \DEQXzero{}+\PO{}, but the instance of \ThreePartition{} does not have a solution. For every $i \in [r]$, the chore $c_{3r+i}$ must be assigned to agent $i$ under any Pareto optimal allocation. Regardless of how the remaining chores $c_1,\dots,c_{3r}$ are allocated, it must be the case that there exist agents $i,k$ with $v_i(A_i) > v_k(A_k)$. All goods of non-zero value owned by an agent are valued at an even number. In particular, the quantities $v_i(A_i)$ and $v_k(A_k)$ are even, and so is their difference. Furthermore, since $v_i(\{c_{3r+k}\}) = -1$, we have that $v_i(A_i \cup \{c_{3r+k}\}) > v_k(A_k)$, implying that $A$ violates \DEQXzero{}---a contradiction.
\end{proof}

We remark that the proof of \Cref{lem:DEQXzero_PO_strongly_NPHard} can also be used to show strong \NPH{}ness of checking the existence of an \EFXzero{}+\PO{} allocation, where \EFXzero{} is the analogue of \EFX{} without the $v_{i,j} < 0$ condition. That is, an allocation $A$ is \EFXzero{} if for every pair of agents $i,k \in [n]$ such that $A_i \neq \emptyset$ and for every chore $j \in A_i$, we have $v_i(A_i \setminus \{j\}) \geq v_i(A_k)$.

In the remainder of this section, we will show that checking the existence of \EQXzero{}+\PO{} allocations remains \NPH{} even for the special case of \emph{binary valuations} (\Cref{thm:EQXzero_Binary_Hardness}). Recall that a chores instance is said to have binary valuations if for every agent $i \in [n]$ and every chore $j \in [m]$, we have $v_{i,j} \in \{-1,0\}$. We will write $\Gamma \coloneqq \{j \in [m] \, : \, v_{i,j} = 0 \text{ for some agent } i \in [n]\}$ to denote the set of chores that are valued at $0$ by one or more agents. It is easy to see that for binary valuations, a necessary and sufficient condition for an allocation to be Pareto optimal is that each chore in $\Gamma$ is assigned to an agent that values it at $0$. As a consequence, for binary valuations, one can check in polynomial time whether a given allocation satisfies Pareto optimality.

\begin{restatable}[\textbf{Hardness of \EQXzero{}+\PO{} for binary valuations}]{theorem}{EQXzeroPOBinaryHardness}
Determining whether a given instance with binary valuations admits an allocation that is equitable up to any possibly zero-valued chore $(\EQXzero{})$ and Pareto optimal $(\PO{})$ is \NPC{}.
\label{thm:EQXzero_Binary_Hardness}
\end{restatable}
\begin{proof}
We will show a reduction from \VertexCover{}, which is known to be \NPC{}~\citep{GJ79computers}. An instance of \VertexCover{} consists of a graph $G = (V,E)$ and a positive integer $k \leq |V|$. The goal is determine whether $G$ admits a \emph{vertex cover} of size at most $k$ (i.e., a set $V' \subseteq V$ such that $|V'|\leq k$ and for every edge $e \in E$, there exists a vertex $v \in V'$ that is adjacent to $e$ in the graph $G$). We will use $r \coloneqq |V|$ and $s \coloneqq |E|$ to denote the number of vertices and edges in the graph $G$, respectively.

We will construct a fair division instance with $r$ \emph{vertex} agents, denoted by $a_1,\dots,a_r$, and $r+s-k$ chores. The set of chores consists of $s$ \emph{edge} chores $C_1,\dots,C_s$ and $r-k$ \emph{dummy} chores $D_1,\dots,D_{r-k}$. Each dummy chore is valued at $-1$ by all agents. Additionally, for every $i \in [r]$ and every $j \in [s]$, the edge chore $C_j$ is valued at $0$ by the vertex agent $a_i$ if the vertex $v_i$ is adjacent to the edge $e_j \in E$ in the graph $G$ (i.e., $v_i \in e_j$), and at $-1$ otherwise.

($\Rightarrow$) Suppose $V' \subseteq V$ is a vertex cover of size at most $k$. Then, the desired allocation, say $A$, can be constructed as follows: For every $j \in [s]$, the edge chore $C_j$ is assigned to vertex agent $a_i$ if the vertex $v_i$ is in the vertex cover and the edge $e_j$ is adjacent to $v_i$, i.e., $v_i \in V'$ and $v_i \in e_j$.\footnote{If multiple vertex agents fit this description, then pick one arbitrarily.} The dummy chores are assigned uniformly among the $r-k$ agents whose corresponding vertices are not included in the vertex cover.

The allocation constructed above satisfies the sufficient condition for Pareto optimality, so we only need to check for \EQXzero{}. Notice that the utility of an agent in the allocation $A$ is either $0$ (for agents who get the edge chores) or $-1$ (for agents who get the dummy chores). For any pair of agents $a_i,a_k$ such that $v_i(A_i) < v_k(A_k)$, it must be that $v_i(A_i) = -1$ and $v_k(A_k) = 0$. Then, by construction, agent $a_i$ gets exactly one (dummy) chore, i.e., $|A_i| = 1$. In that case, we have that for every chore $c \in A_i$, $v_i(A_i \setminus \{c\}) = 0 \geq v_k(A_k)$, implying that $A$ is \EQXzero{}.

($\Leftarrow$) Now suppose there exists an \EQXzero{}+\PO{} allocation $A$. Then, $A$ must satisfy the necessary condition for Pareto optimality, that is, any chore that is valued at zero by one or more agents must be assigned to some agent that values it at $0$. Thus, any edge chore $C_j$ is assigned to a vertex agent $a_i$ such that $v_i \in e_j$. This, in turn, means that if an agent does not get a dummy chore, then its utility must be $0$.

We claim that no agent can be assigned more than one dummy chore. Suppose, for contradiction, that agent $a_i$ gets two or more dummy chores. Then, $v_i(A_i) \leq -2$. Since the number of dummy chores is strictly smaller than the number of agents, some agent, say $a_k$, must miss out on getting a dummy chore. Then, by the above observation, we must have that $v_k(A_k) = 0$. This, however, creates a violation of \EQXzero{}, since $v_i(A_i \setminus \{c\}) \leq -1$ for every chore $c \in A_i$ because of binary valuations. Therefore, each dummy chore must be assigned to a distinct agent. We will write $N_D$ to denote the set of all agents that are assigned a dummy chore in $A$. Notice that $|N_D| = r-k$.

We will now argue that no agent in $N_D$ is assigned an edge chore. Indeed, if some agent $a_i \in N_D$ gets an edge chore $C_j$, then from \EQXzero{} condition, we have that for any other agent $a_k$, 
\begin{equation}
    v_i(A_i \setminus \{C_j\}) \geq v_k(A_k).
    \label{eqn:EQXzero_proof}
\end{equation}
In particular, \Cref{eqn:EQXzero_proof} should hold for any $a_k \in [n] \setminus N_D$. By the above observation, all edge chores are assigned to agents that value them at $0$, thus $v_i(\{C_j\}) = 0$. Substituting this in \Cref{eqn:EQXzero_proof} gives that $v_i(A_i) \geq v_k(A_k)$, which is a contradiction because $v_i(A_i) = -1$ ($a_i$ gets a dummy chore) and $v_k(A_k) = 0$ ($a_k$ does not get any dummy chore). Thus, all edge chores must be assigned to agents in the set $[n] \setminus N_D$, which consists of $k$ agents.

Define $V' \coloneqq \{v_i \in V : a_i \text{ gets an edge chore in } A\}$ as the set of vertices of graph $G$ whose corresponding vertex agents get one or more edge chores in $A$. By the above argument, we have that $|V'| \leq k$. Furthermore, for every edge $e_j \in E$, if the edge chore $C_j$ is assigned to vertex agent $a_i$, then $v_{i,j} = 0$, meaning that the vertex $v_i$ must be adjacent to edge $e_j$. Thus, $V'$ is a valid vertex cover, as desired.
\end{proof}

\subsection{Additional Experimental Details}
\label{subsec:Experiments_Appendix}

\paragraph{Synthetic dataset}
For synthetic experiments, we work with five agents and twenty chores. Following the experimental setup of \citet{FSV+19equitable}, we generate $1000$ instances with valuations drawn from the Dirichlet distribution, with the concentration parameter set to $10$. For each agent, the sampling process returns $m$ positive real numbers that sum up to $1$. These numbers are multiplied by a common `budget' (equal to $1000$). We then round \emph{up} these numbers (in order to ensure integrality) and flip their sign (for chores). Finally, in order to restore normalization, a good is selected uniformly at random and the agent's valuation for that good is increased by $1$. This process is repeated until the (integral) valuations sum up to the budget (i.e., become normalized).

\paragraph{Spliddit dataset}
As mentioned previously in \Cref{sec:Experiments}, Spliddit allows the users to specify the number of copies of each chore, and express their preferences in the form of multipliers. The website further assumes that each agent's valuation for the grand bundle (i.e., all copies of all chores) is $-100$, and uses this assumption to infer the exact numerical preferences of the agents from the multipliers.

We perform one rounding step and two pruning steps on the dataset obtained from Spliddit, as follows:

\begin{itemize}
    \item We round \emph{down} (i.e., away from zero) each agent's valuation for each copy of every chore to ensure integrality. Note that the rounded valuations might no longer be normalized.

    \item We remove the instances with fewer chores than agents (i.e., $m<n$), since any matching of chores to agents satisfies \EQ{1}/\EQX{} for these instances.

    \item We remove instances with more than $1000$ copies of any chore. This choice is made because the aforementioned rounding step typically results in instances where all copies of all chores are valued at $-1$ for every agent, leading to trivial instances.
\end{itemize}

The aforementioned pruning steps bring down the number of instances from $2889$ to $2613$. As mentioned earlier, in the pruned dataset, the number of agents ranges between $2$ and $15$, and the number of chore ranges between $3$ and $1100$.

\paragraph{Spliddit algorithm}
Spliddit implements the following two-step randomized algorithm:\footnote{http://www.spliddit.org/apps/tasks}

\begin{itemize}
    \item First, a linear program is used to compute an \emph{egalitarian-equivalent} solution~\citep{PS78egalitarian}, i.e., a fractional allocation that is perfectly equitable and minimizes the agents' disutilities. From a generalization of the Birkhoff-von Neumann theorem~\citep{BCK+13designing}, it follows that this fractional solution can be expressed a distribution over integral allocations wherein all copies of a chore are assigned to the same agent. The said fractional allocation is therefore ex ante equitable.

    \item Next, another linear program is used to construct the aforementioned distribution, and subsequently an integral allocation is sampled from this distribution so that the expected number of chores assigned to each agent is preserved.\\
\end{itemize}

\end{document}